\documentclass[letterpaper,11pt]{article}
\pdfoutput=1
\usepackage[letterpaper,margin=1in]{geometry}
\usepackage[utf8]{inputenc}
\usepackage[english]{babel}
\usepackage[numbers,sort&compress]{natbib}
\usepackage{microtype}
\usepackage[small]{caption}
\usepackage{xcolor}
\usepackage{graphicx}
\usepackage{amsmath}
\usepackage{amsthm}
\usepackage{amssymb}
\usepackage{algorithm2e}
\usepackage{enumitem}
\usepackage{booktabs}
\usepackage{bm}
\usepackage{calc}
\usepackage[unicode]{hyperref}
\usepackage{doi}

\hypersetup{
    colorlinks=true,
    linkcolor=black,
    citecolor=black,
    filecolor=black,
    urlcolor=[rgb]{0,0.1,0.5},
    pdftitle={Classification of distributed binary labeling problems},
    pdfauthor={Alkida Balliu, Sebastian Brandt, Yuval Efron, Juho Hirvonen, Yannic Maus, Dennis Olivetti, Jukka Suomela}
}

\newtheorem{theorem}{Theorem}[section]
\newtheorem{lemma}[theorem]{Lemma}
\newtheorem{claim}[theorem]{Claim}
\newtheorem{corollary}[theorem]{Corollary}
\newtheorem{observation}[theorem]{Observation}

\theoremstyle{definition}
\newtheorem{definition}[theorem]{Definition}
\newtheorem{example}[theorem]{Example}
\theoremstyle{remark}
\newtheorem*{remark}{Remark}
\newtheoremstyle{repeatdefinition}{\topsep}{\topsep}{}{}{\bfseries}{.}{ }{\thmname{#1}\thmnote{ \bfseries #3}}
\theoremstyle{repeatdefinition}
\newtheorem{examplerep}{Example}

\newcommand{\set}[1]{\left\{ #1 \right\}}

\DeclareMathOperator{\black}{black}
\DeclareMathOperator{\white}{white}
\DeclareMathOperator{\indegree}{indegree}
\DeclareMathOperator{\outdegree}{outdegree}
\DeclareMathOperator{\poly}{poly}
\DeclareMathOperator{\lowdegree}{low-degree}
\DeclareMathOperator{\longpaths}{long-paths}
\DeclareMathOperator{\RC}{RC}
\DeclareMathOperator{\RCP}{RCP}

\newcommand{\type}[1]{\ensuremath{\color{black!70}\bm{\mathsf{#1}}}}
\newcommand{\family}[1]{\ensuremath{\mathsf{#1}}}

\newcommand{\tI}{\type{I}}
\newcommand{\tII}{\type{II}}
\newcommand{\tIII}{\type{III}}
\newcommand{\tIV}{\type{IV}}
\newcommand{\tV}{\type{V}}
\newcommand{\tVI}{\type{VI}}
\newcommand{\tVII}{\type{VII}}
\newcommand{\fIa}{\family{I.a}}
\newcommand{\fIb}{\family{I.b}}
\newcommand{\fIc}{\family{I.c}}
\newcommand{\fId}{\family{I.d}}
\newcommand{\fIIa}{\family{II.a}}
\newcommand{\fIIb}{\family{II.b}}
\newcommand{\fIIIa}{\family{III.a}}
\newcommand{\fIIIb}{\family{III.b}}
\newcommand{\fIVa}{\family{IV.a}}
\newcommand{\fIVb}{\family{IV.b}}
\newcommand{\fVa}{\family{V.a}}
\newcommand{\fVb}{\family{V.b}}
\newcommand{\fVIa}{\family{VI.a}}
\newcommand{\fVIb}{\family{VI.b}}
\newcommand{\fVIIa}{\family{VII.a}}

\newcommand{\cA}{\mathcal{A}}

\newcommand{\cF}{\mathcal{F}}

\newcommand{\fdso}{\ensuremath{\mathsf{FDSO}}}
\newcommand{\sfdso}{\ensuremath{\mathsf{FDSO}(s)}}

\newcommand{\bA}{\mathbf{A}}
\newcommand{\bH}{\mathbf{H}}
\newcommand{\bT}{\mathbf{T}}
\newcommand{\bX}{\mathbf{X}}
\newcommand{\bY}{\mathbf{Y}}
\newcommand{\bZ}{\mathbf{Z}}

\newcommand{\rawp}{1}
\newcommand{\rawn}{0}
\newcommand{\rawx}{\makebox[\widthof{0}][c]{$\ast$}}
\newcommand{\pcol}{\color{green!50!black}}
\newcommand{\ncol}{\color{red!50!black}}
\newcommand{\oneormore}{\makebox[\widthof{0}]{\hspace{-1.5pt}$\scriptscriptstyle\boldsymbol+$}}
\newcommand{\n}{{\ncol\rawn}}
\newcommand{\p}{{\pcol\rawp}}
\newcommand{\x}{\rawx}
\newcommand{\nnn}{{\ncol\rawn\raisebox{5.5pt}[0pt]{\oneormore}}}
\newcommand{\ppp}{{\pcol\rawp\raisebox{4.5pt}[0pt]{\oneormore}}}
\newcommand{\xxx}{{\rawx\raisebox{4pt}[0pt]{\oneormore}}}

\newcommand{\myemail}[1]{\,$\cdot$\, {\small #1}}
\newcommand{\myaff}[1]{\,$\cdot$\, {\small #1}\par\smallskip}

\newenvironment{myabstract}
{\list{}{\listparindent 1.5em%
        \itemindent    \listparindent
        \leftmargin    1cm
        \rightmargin   1cm
        \parsep        0pt}%
    \item\relax}
{\endlist}

\newenvironment{mycover}
{\list{}{\listparindent 0pt
        \itemindent    \listparindent
        \leftmargin    1cm
        \rightmargin   1cm
        \parsep        0pt}%
    \raggedright
    \item\relax}
{\endlist}

\begin{document}

\begin{mycover}
{\huge\bfseries\boldmath Classification of distributed \\ binary labeling problems \par}
\bigskip
\bigskip
\bigskip

\textbf{Alkida Balliu}
\myemail{alkida.balliu@aalto.fi}
\myaff{Aalto University}

\textbf{Sebastian Brandt}
\myemail{brandts@ethz.ch}
\myaff{ETH Zurich}

\textbf{Yuval Efron}
\myemail{szxrtde@cs.technion.ac.il}
\myaff{Technion}

\textbf{Juho Hirvonen}
\myemail{juho.hirvonen@aalto.fi}
\myaff{Aalto University}

\textbf{Yannic Maus}
\myemail{yannic.maus@campus.technion.ac.il}
\myaff{Technion}

\textbf{Dennis Olivetti}
\myemail{dennis.olivetti@aalto.fi}
\myaff{Aalto University}

\textbf{Jukka Suomela}
\myemail{jukka.suomela@aalto.fi}
\myaff{Aalto University}
\end{mycover}
\bigskip

\begin{myabstract}
\noindent\textbf{Abstract.}
We present a complete classification of the deterministic distributed time complexity for a family of graph problems: \emph{binary labeling problems} in trees. These are locally checkable problems that can be encoded with an alphabet of size two in the edge labeling formalism. Examples of binary labeling problems include sinkless orientation, sinkless and sourceless orientation, 2-vertex coloring, perfect matching, and the task of coloring edges red and blue such that all nodes are incident to at least one red and at least one blue edge. More generally, we can encode e.g.\ any cardinality constraints on indegrees and outdegrees.

We study the deterministic time complexity of solving a given binary labeling problem in trees, in the usual LOCAL model of distributed computing. We show that the complexity of any such problem is in one of the following classes: $O(1)$, $\Theta(\log n)$, $\Theta(n)$, or unsolvable. In particular, a problem that can be represented in the binary labeling formalism cannot have time complexity $\Theta(\log^* n)$, and hence we know that e.g.\ any encoding of maximal matchings has to use at least three labels (which is tight).

Furthermore, given the description of any binary labeling problem, we can easily determine in which of the four classes it is and what is an asymptotically optimal algorithm for solving it. Hence the distributed time complexity of binary labeling problems is \emph{decidable}, not only in principle, but also \emph{in practice}: there is a simple and efficient algorithm that takes the description of a binary labeling problem and outputs its distributed time complexity.

\end{myabstract}

\thispagestyle{empty}
\setcounter{page}{0}
\newpage

\section{Introduction}

This work presents a complete classification of the deterministic distributed time complexity for a family of distributed graph problems: \emph{binary labeling problems} in trees. These are a special case of widely-studied locally checkable labeling problems \cite{Naor1995}. The defining property of a binary labeling problem is that it can be encoded with an \emph{alphabet of size two} in the \emph{edge labeling formalism}, which is a modern representation for locally checkable graph problems \cite{Balliu2019,Brandt2019automatic,Olivetti2019}; we will give the precise definition in Section~\ref{ssec:binary-labeling-problems}.

\paragraph{Contributions.}
In this work, we focus on \emph{deterministic} distributed algorithms in the LOCAL model of distributed computing, and we study the computational complexity of solving a binary labeling problem in \emph{trees}. It is easy to see that there are binary labeling problems that fall in each of the following classes:
\begin{itemize}[noitemsep]
    \item Trivial problems, solvable in $O(1)$ rounds.
    \item Problems similar to sinkless orientation, solvable in $\Theta(\log n)$ rounds \citep{Brandt2016,chang16exponential,ghaffari17distributed}.
    \item Global problems, requiring $\Theta(n)$ rounds.
    \item Unsolvable problems.
\end{itemize}

We show that this is a \emph{complete} list of all possible complexities. In particular, there are no binary labeling problems of complexities such as $\Theta(\log^* n)$ or $\Theta(\sqrt{n})$. For example, maximal matching is a problem very similar in spirit to binary labeling problems, it has a complexity $\Theta(\log^* n)$ in bounded-degree graphs \cite{Linial1992,cole86deterministic}, and it can be encoded in the edge labeling formalism using an alphabet of size three \cite{Balliu2019}---our work shows that three labels are also necessary for all problems in this complexity class.

Moreover, using our results one can easily determine the complexity class of any given binary labeling problem. We give a simple, concise characterization of all binary labeling problems for classes $O(1)$, $\Theta(n)$, and unsolvable, and we show that all other problems belong to class $\Theta(\log n)$. Hence the deterministic distributed time complexity of a binary labeling problem is \emph{decidable}, not only in theory but also \emph{in practice}: given the description of any binary labeling problem, a human being or a computer can easily find out the distributed computational complexity of the problem, as well as an asymptotically optimal algorithm for solving the problem. Our classification of all binary labeling problems is presented in Table~\ref{tab:deterministic}, and given any binary labeling problem $\Pi$, one can simply do mechanical pattern matching to find its complexity class in this table.

Our work also sheds new light on the \emph{automatic round elimination technique} \cite{Brandt2019automatic,Olivetti2019}. Previously, it was known that sinkless orientation is a nontrivial \emph{fixed point} for round elimination \cite{Brandt2016}---such fixed points are very helpful for lower bound proofs, but little was known about the existence of other nontrivial fixed points. Our classification of binary labeling problems in this work led to the discovery of new nontrivial fixed points---this will hopefully pave the way for the development of a theoretical framework that enables us to understand when round elimination leads to fixed points and why.

This work will also make it \emph{easier to prove lower bounds in the future}. Before this work, in essence the only known way to prove a nontrivial $\Omega(\log n)$ lower bound for some locally checkable problem $\Pi$ was to come up with a reduction showing that $\Pi$ is at least as hard as sinkless orientation. Our systematic exploration of binary labeling problems led to the discovery of an entire family of simple graph problems that are not directly comparable with sinkless orientation, but for which we can prove tight $\Omega(\log n)$ lower bounds directly with round elimination. In the future, whenever we encounter a graph problem $\Pi$ of an unknown complexity, we can try to relate it not only to sinkless orientation but also to one of the new problems for which we now have new lower bounds.

\paragraph{Open questions.}
The main open question that we leave for future work is extending the characterization to randomized distributed algorithms: some binary labeling problems can be solved in $\Theta(\log \log n)$ rounds with randomized algorithms, but it is not yet known exactly which binary labeling problems belong to this class. Our work takes the first steps towards developing such a classification.

\paragraph{Structure.}
We start with a brief discussion of the general landscape of distributed computational complexity in Section~\ref{sec:background}, and then give formal definitions of binary labeling problems in Section~\ref{ssec:binary-labeling-problems}. We present a summary of our results in Section~\ref{sec:contributions}---our main contribution is the characterization of all binary labeling problems in Table~\ref{tab:deterministic}. We explain the details of the model of computing in Section~\ref{sec:model}. All algorithms and lower bound proofs related to deterministic complexity are presented in Sections \ref{sec:unsol}--\ref{sec:log-lower}:
\begin{itemize}[noitemsep]
    \item Section~\ref{sec:unsol}: unsolvable problems,
    \item Section~\ref{sec:constant}: complexity class $O(1)$,
    \item Section~\ref{sec:global}: complexity class $\Theta(n)$,
    \item Section~\ref{sec:log-upper}: $O(\log n)$ upper bounds,
    \item Section~\ref{sec:log-lower}: $\Omega(\log n)$ lower bounds.
\end{itemize}
We conclude with a discussion of randomized complexity in Section~\ref{sec:rand}.

\section{Background and related work}\label{sec:background}

\paragraph{Locality and the LOCAL model.}

A fundamental concept in distributed computing is \emph{locality}. In brief, the locality of a graph problem is $T$ if all nodes of a graph can produce their own part of the solution by only looking at their radius-$T$ neighborhoods.

The idea of locality is formalized as the time complexity in the \emph{LOCAL model} of distributed computing \cite{Peleg2000,Linial1992}. In brief, a graph problem $\Pi$ can be solved in graph family $\cF$ in time $T(n)$ in the LOCAL model if there is a mapping $\cA$ from radius-$T(n)$ neighborhoods to local outputs such that for any graph $G \in \cF$ with $n$ nodes and for any assignment of polynomially-sized unique identifiers in $G$ we can apply mapping $\cA$ to each node of $G$, and the local outputs form a feasible solution for $\Pi$. Here $\cA$ is a deterministic distributed algorithm in the LOCAL model.

Note that if we interpret graph $G$ as a computer network in which nodes can send messages to their neighbors, then time becomes equal to distance: in $T$ synchronous communication rounds all nodes can gather their radius-$T$ neighborhoods (and nothing more). Hence we can interchangeably refer to locality, \emph{distributed time complexity}, and the \emph{number of communication rounds}.

Our main focus is on deterministic distributed algorithms, but we will also discuss randomized distributed algorithms; the only difference there is that in a randomized algorithm the nodes are labeled not only with unique identifiers but also with an unbounded stream of random bits. We will define the model of computing in more detail in Section~\ref{sec:model}.

\paragraph{Distributed complexity theory and LCL problems.}

The study of distributed graph algorithms has traditionally focused on specific graph problems---for example, investigating exactly what is the locality of finding a maximal independent set. However, in the recent years we have seen more focus on the development of a distributed complexity theory with which we can reason about entire \emph{families of graph problems} \cite{Balliu2018disc,Balliu2018stoc,Balliu2019decidable,Balliu2019padding,Brandt2016,Brandt2017,chang16exponential,Chang2019,fischer17sublogarithmic,Ghaffari2017,Ghaffari2018,Korhonen2018,Rosenbaum2019,Rozhon2019}.

The key example is the family of \emph{locally checkable labeling} problems (LCLs), introduced by \citet{Naor1995}. Informally, a problem is locally checkable if the feasibility of a solution can be verified by looking at all constant-radius neighborhoods. For example, maximal independent sets are locally checkable, as we can verify both independence and maximality by looking at radius-$1$ neighborhoods.

In this line of research, among the most intriguing results are various \emph{gap theorems}: For example, there are LCL problems solvable in $\Theta(\log^* n)$ rounds, while some LCL problems require $\Theta(\log n)$ rounds. However, between these two classes there is a gap: there are no LCL problems whose deterministic complexity in bounded-degree graphs is between $\omega(\log^* n)$ and $o(\log n)$ \cite{chang16exponential}.

\paragraph{Decidability of distributed computational complexity.}

The existence of such a gap immediately suggests a follow-up question: given the description of an LCL problem, can we \emph{decide} on which side of the gap it lies? And if so, can we automatically construct an asymptotically optimal algorithm for solving the problem?

As soon as we look at a family of graphs that contains e.g.\ 2-dimensional grids, questions related to the distributed complexity of a given LCL problem become undecidable \cite{Naor1995,Brandt2017}. However, if we look at the case of paths and cycles, we can at least in principle write a computer program that determines the computational complexity of a given LCL problem \cite{Naor1995,Brandt2017,Balliu2019decidable}, and some questions related to the complexity of LCLs in trees are also decidable \cite{Chang2019}---unfortunately, we run into PSPACE-hardness already in the case of paths and cycles \cite{Balliu2019decidable}.

We conjecture that \emph{all} questions about the distributed complexity of LCL problems in trees are decidable. Proving (or disproving) the conjecture is a major research program, but in this work we take one step towards proving the conjecture and we bring plenty of good news: we introduce a family of LCL problems, so-called \emph{binary labeling problems}, and we show that we can completely characterize the deterministic distributed complexity of every binary labeling problem in trees. In particular, all questions about the deterministic distributed complexity of these problems are decidable not only in principle but also in practice---using our results, a human being or a computer can easily find an optimal algorithm for solving any given binary labeling problem.

\section{Binary labeling problems}\label{ssec:binary-labeling-problems}

We will now give the formal definition of the family of binary labeling problems, together with some natural examples of graph problems that are contained in this family. LCL problems have been traditionally specified by listing a \emph{collection of permitted local neighborhoods} \cite{Naor1995}. However, we will use a more recent \emph{edge labeling} formalism \cite{Balliu2019,Brandt2019automatic,Olivetti2019}, which is equally expressive in the case of trees, and it has the additional benefit that it makes it very convenient to apply the automatic round elimination technique \cite{Brandt2019automatic,Olivetti2019}:
\begin{itemize}[noitemsep]
    \item we have a bipartite graph and the nodes are colored with two colors, white and black,
    \item the task is to \emph{label edges} with symbols from some alphabet $\Sigma$,
    \item there are \emph{both white and black constraints} that define the graph problem.
\end{itemize}
Now if we set $|\Sigma| = 2$, we arrive at the definition of binary labeling problems.

\subsection{General form}

A \emph{binary labeling problem} is a tuple $\Pi = (d,\delta,W,B)$, where
\begin{itemize}[noitemsep]
\item $d \in \{2,3,\dotsc\}$ is the \emph{white degree},
\item $\delta \in \{2,3,\dotsc\}$ is the \emph{black degree},
\item $W \subseteq \{0,1,\dotsc,d\}$ is the \emph{white constraint}, and
\item $B \subseteq \{0,1,\dotsc,\delta\}$ is the \emph{black constraint}.
\end{itemize}

An \emph{instance} of problem $\Pi$ is a pair $(G,f)$, where
\begin{itemize}[noitemsep]
\item $G = (V,E)$ is a simple graph,
\item $f\colon V \to \{ \black, \white \}$ is a proper $2$-coloring of $G$, i.e., $f(u) \ne f(v)$ for all $\{u,v\} \in E$.
\end{itemize}

Let $X \subseteq E$ be a subset of edges. For each node $v \in V$, its $X$-degree $\deg_X(v)$ is the number of edges in $X$ that are incident to $v$. We say that $X \subseteq E$ is a \emph{solution} to binary labeling problem $\Pi$ if it satisfies the following constraints for all $v \in V$:
\begin{itemize}[noitemsep]
\item If $f(v) = \white$ and $\deg(v) = d$, then $\deg_X(v) \in W$.
\item If $f(v) = \black$ and $\deg(v) = \delta$, then $\deg_X(v) \in B$.
\end{itemize}
The interpretation is that problem $\Pi$ is interesting in regular neighborhoods in which all white nodes have degree $d$ and all black nodes have degree $\delta$; any irregularities make the problem easier to solve, as all other nodes are unconstrained. We use the term \emph{relevant nodes} to refer to white nodes of degree $d$ and black nodes of degree $\delta$.

\begin{example}[bipartite splitting]\label{ex:bipartite-splitting}
Let $d = \delta = 4$ and $W = B = \{1,2,3\}$. We can interpret a solution $X \subseteq E$ as a coloring: edges in $X$ are colored red and all other edges are colored blue. Now $\Pi$ is equivalent to the following graph problem on bipartite graphs: color all edges red or blue such that all degree-$4$ nodes are incident to at least one blue edge and at least one red edge.
\end{example}

\subsection{Equivalence}
\label{sssec:equivalence}

If $a$ is a natural number and $A$ is a set of natural numbers, we will define
\[
    a - A = \{ a - x : x \in A \}.
\]

\begin{observation}\label{obs:equiv}
The following binary labeling problems have the same distributed complexity up to $\pm 1$ rounds:
\begin{align*}
\Pi_{00} &= (d,\delta,W,B), \\
\Pi_{01} &= (\delta,d,B,W), \\
\Pi_{10} &= (d,\delta,d-W,\delta-B), \\
\Pi_{11} &= (\delta,d,\delta-B,d-W).
\end{align*}
\end{observation}
\begin{proof}
Given an algorithm for $\Pi_{x0}$, we get an algorithm for $\Pi_{x1}$ by exchanging the roles of black and white nodes. Given an algorithm for $\Pi_{0x}$, we get an algorithm for $\Pi_{1x}$ by replacing solution $X$ with its complement $E \setminus X$.

When we exchange the roles of black and white nodes, we may need to spend one additional round so that black nodes can inform white nodes of the final output or vice versa; see Section~\ref{sec:model} and Observation~\ref{obs:white-black-algorithm} for more details.
\end{proof}
Observation~\ref{obs:equiv} partitions binary labeling problems in equivalence classes, each of them with at most four distinct problems. We use notation $\Pi \sim \Pi'$ for this equivalence relation, and say that $\Pi$ and $\Pi'$ are \emph{equivalent}.

\subsection{Restrictions and relaxations}

\begin{definition}\label{def:restriction}
Given two problem $\Pi = (d,\delta,W,B)$ and $\Pi' = (d,\delta,W',B')$, we say that $\Pi'$ is a \emph{restriction} of $\Pi$ and $\Pi$ is a \emph{relaxation} of $\Pi'$ if
\[
W' \subseteq W, \quad B' \subseteq B.
\]
\end{definition}

We use notation $\Pi' \subseteq \Pi$ to denote that $\Pi'$ is a restriction of $\Pi$.

\begin{observation}\label{obs:restriction}
If $\Pi' \subseteq \Pi$, then any feasible solution for $\Pi'$ is also a feasible solution for $\Pi$. In particular, if $\Pi'$ can be solved in $T$ rounds, then $\Pi$ can also be solved in $T$ rounds.
\end{observation}

\subsection{Vector notation}

It is convenient to interpret set $W$ as a bit vector $w_0w_1\dots w_d$ with $d+1$ bits, so that bit $w_i = \p$ if $i \in W$ and $w_i = \n$ if $i \notin W$. Similarly, $B$ can be interpreted as a bit vector with $\delta+1$ bits.

\begin{examplerep}[\ref{ex:bipartite-splitting}b]
Using this notation, the problem of Example~\ref{ex:bipartite-splitting} can be represented as follows:
\[
    W = \n\p\p\p\n, \quad
    B = \n\p\p\p\n,
\]
or, in brief, $\Pi = (\n\p\p\p\n, \n\p\p\p\n)$.
\end{examplerep}

Note that when we use vector notation, vectors $W$ and $B$ fully determine problem $\Pi$; therefore we do not need to specify $d$ and $\delta$ separately and we can simply write $\Pi = (W,B)$.

We will use shorthand notation such as $\p^{x}$ for a vector of $x$ $\p$s and $\ppp$ for a vector of one or more $\p$s, and we will use $\x$ to refer to a bit of any value. For example, $W = \n\x\x\n$ is a shorthand for $W \in \{\n\n\n\n,\allowbreak\n\n\p\n,\allowbreak\n\p\n\n,\allowbreak\n\p\p\n\}$ and $W = \n\ppp\n$ is a shorthand for $W \in \{\n\p\n,\allowbreak\n\p\p\n,\allowbreak\n\p\p\p\n,\dotsc\}$.

\subsection{Important special case: \texorpdfstring{\boldmath$\delta=2$}{delta = 2}}\label{ssec:case-delta2}

At first the bipartite setting may seem restrictive---indeed, why would we care about graphs that are properly $2$-colored. However, we can take any graph and interpret edges as ``black nodes'' and this way many graph problems of interest can be represented in the binary labeling formalism.

More precisely, let $G_0 = (V_0, E_0)$ be a graph. Subdivide all edges of $G_0$ to construct a new graph $G = (V,E)$, and assign the color white to all original nodes in $V_0$ and color black to the new nodes in $V \setminus V_0$.

\begin{figure}
\centering
\includegraphics[page=1,scale=0.9]{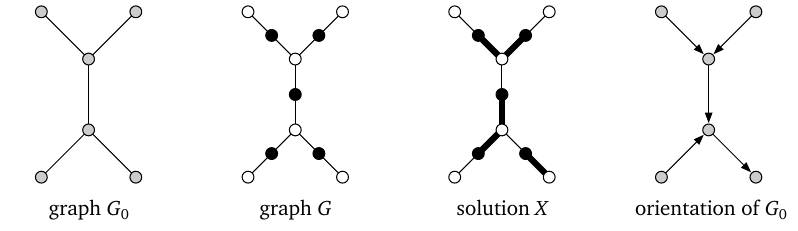}
\caption{Binary labeling problems with $B = \n\p\n$ are orientation problems.}\label{fig:delta2orientation}
\end{figure}

Now consider a binary labeling problem $\Pi$ with $B = \n\p\n$. Assume that we have a solution $X$ to $\Pi$ in $G$. We can now interpret $X$ as an \emph{orientation} of the original graph $G_0$: If an edge $e = \{u,v\} \in E_0$ was subdivided in two edges, $e_1 = \{u,x\}$, and $e_2 = \{x,v\}$, we will have exactly one of these edges in set $X$. If we have $e_1 \in X$, we can interpret it so that edge $\{u,v\}$ is oriented from $v$ to $u$, and otherwise it is oriented from $u$ to $v$; see Figure~\ref{fig:delta2orientation}. In essence, $\Pi$ is now equivalent to the following problem: Find an orientation of $G_0$ such that all nodes with $\deg(v) = d$ have $\indegree(v) \in W$.

\begin{example}[sinkless orientation]\label{ex:sinkless-orientation}
Let $d > 2$, $\delta = 2$, $W = \p\p\ppp\n$, and $B = \n\p\n$. Now all edges must be properly oriented: there is exactly one head. Furthermore, for all nodes of degree $d$, they must have indegree in $\{0,1,\dotsc,d-1\}$. Put otherwise, degree-$d$ nodes must have outdegree at least one. Hence a feasible solution represents an orientation in which none of degree-$d$ nodes are sinks.
\end{example}

\begin{examplerep}[\ref{ex:sinkless-orientation}b]
Let $W = \n\p\p\ppp$ and $B = \n\p\n$. By Observation~\ref{obs:equiv}, this is equivalent to Example~\ref{ex:sinkless-orientation}. In essence, we have merely reversed the roles of heads and tails---now the task is to find a sourceless orientation.
\end{examplerep}

\begin{example}[sinkless and sourceless orientation]\label{ex:sinkless-sourceless}
Let $W = \n\p\ppp\n$ and $B = \n\p\n$. Now in $d$-regular graphs the task is to find an orientation such that all nodes have at least one incoming and at least one outgoing edge.

Note that this problem is a restriction of Example~\ref{ex:sinkless-orientation}; recall Definition~\ref{def:restriction}.
\end{example}

\begin{example}[even orientation]\label{ex:even-orientation}
Let $W = \p\n\p\n$ and $B = \n\p\n$. In $3$-regular graphs the task is to find an orientation such that all nodes have an even number of outgoing edges.
\end{example}

\begin{example}[two-coloring]\label{ex:two-coloring}
Let $W = \p\nnn\p$ and $B = \n\p\n$. In $d$-regular graphs this is a $2$-coloring problem: each node $v$ is either ``red'' (indegree $0$) or ``blue'' (indegree $d$), and all edges between such nodes are properly colored (they are always oriented from red to blue).
\end{example}

\begin{figure}
\centering
\includegraphics[page=2,scale=0.9]{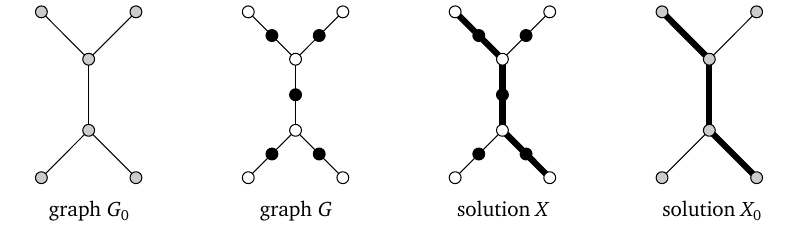}
\caption{Binary labeling problems with $B = \p\n\p$ are splitting problems.}\label{fig:delta2consistent}
\end{figure}

Another interesting special case is $B = \p\n\p$. Now for each original edge of $G_0$ we will select either both of the half-edges or none of the half-edges, and hence a solution $X \subseteq E$ can be interpreted in a natural way as a \emph{subset of edges} $X_0 \subseteq E_0$ in the original graph; see Figure~\ref{fig:delta2consistent}. This is a splitting problem: partition $E_0$ in two classes, $X_0$ and $E_0 \setminus X_0$, and $W$ determines how many incident edges in each class we can have.

\begin{example}[regular matching]\label{ex:regular-matching}
Let $W = \n\p\n\nnn$ and $B = \p\n\p$. 
The task is to find a set $X_0 \subseteq E_0$ such that all nodes of degree $d$ are incident to exactly one edge in $X_0$. In particular, if we have a $d$-regular graph, $X_0$ is a perfect matching.
\end{example}

\begin{example}[splitting]\label{ex:splitting}
Let $W = \n\p\ppp\n$ and $B = \p\n\p$. In this problem we will need to color edges red and blue such that all degree-$d$ nodes are incident to at least one red edge and at least one blue edge.
\end{example}

\begin{figure}
\centering
\includegraphics[page=3,scale=0.9]{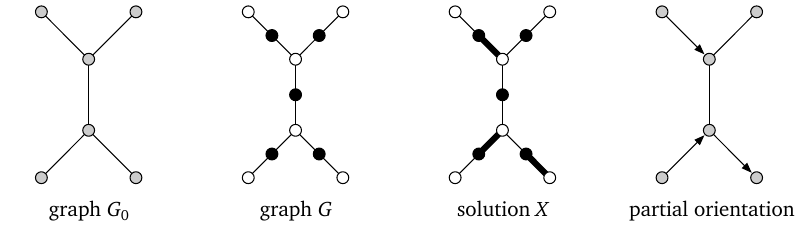}
\caption{Binary labeling problems with $B = \p\p\n$ are partial orientation problems.}\label{fig:delta2partial}
\end{figure}

We can also consider e.g.\ $B = \p\p\n$ and interpret $X$ as a \emph{partial orientation}: some edges of $G_0$ are oriented and some may be left unoriented and again $W$ indicates which indegrees are permitted; see Figure~\ref{fig:delta2partial}. The case of $B = \n\p\p$ is equivalent to $B = \p\p\n$---recall Observation~\ref{obs:equiv}.

As we will see, all other cases $B = \n\n\n$, $B = \p\n\n$, $B = \n\n\p$, and $B = \p\p\p$ are either trivial or unsolvable; we will give two examples:

\begin{example}[contradiction]\label{ex:contradiction}
Let $W = \n\p\ppp$ and $B = \p\n\n$. Here black nodes must have $X$-degree $0$, so we must have $X = \emptyset$. However, all white nodes must be adjacent to at least one edge in $X$. Hence there is no solution to the problem.
\end{example}

\begin{example}[trivial]\label{ex:trivial}
Let $W = \n\n\p\nnn$ and $B = \p\p\p$. Here all white nodes can arbitrarily choose two incident edges and add them to $X$.
\end{example}

\begin{remark}
Similar ideas can be generalized to hypergraphs. In essence, we can interpret the bipartite graph $G$ as a hypergraph $H$, where white nodes of $G$ correspond to nodes of $H$ and black nodes of $G$ correspond to hyperedges of $H$. Now $\Pi$ is in essence a hypergraph problem. If we set $B = \n\p\nnn$, a solution $X$ can be interpreted as an orientation of hyperedges, and if we set $B = \p\nnn\p$, a solution $X$ can be interpreted as a subset of hyperedges.
\end{remark}

\subsection{Our focus: trees}

Throughout this work, we will study binary labeling problems in trees. Binary labeling problems as such make sense also in general graphs, but as our aim is at understanding LCL problems on trees, we will assume from now on that our input graph is a tree (a simple, connected, acyclic graph).

All of the above examples are well-defined also in trees; however, some care is needed when we interpret them. For example, let us revisit the ``regular matching'' problem from Example~\ref{ex:regular-matching}. Let $W = \n\p\n\nnn$ and $B = \p\n\p$. Now in the case of a tree, the task is to find a subset of edges $X$ such that internal nodes of degree $d$ are incident to exactly one such edge. However, leaf nodes are unconstrained. Informally, if we are in the middle of a $d$-regular tree, we will need to find a solution that locally looks like a perfect matching, but near leaf nodes and other irregularities the output is more relaxed. One consequence is that for this problem a solution always exists (while there are of course trees in which a perfect matching does not exist).

\subsection{Why exactly this family?}\label{ssec:why}

As we have discussed above, our goal here is to initiate a systematic investigation of the distributed complexity of LCL problems on trees, and in this work we will focus on a specific family of LCL problems, binary labeling problems. We will now briefly discuss reasons that led to the choice of this specific family of LCL problems for our work.

\paragraph{Sinkless orientation is a binary labeling problem.}

Sinkless orientation (recall Example~\ref{ex:sinkless-orientation}) is one of the cornerstones of the modern theory of distributed computational complexity. This problem was introduced in our context in 2016 \cite{Brandt2016}, and in Google Scholar there are already more than 30 papers written since 2016 that mention the term ``sinkless orientation'', all of them related to the theory of distributed computing. There are many variants of the definition, but for our purposes all of them are in essence equivalent to each other; we will use the following version: there is a fixed parameter $d > 2$, and the task is to orient edges so that nodes of degree $d$ are not sinks (note that low-degree nodes can be sinks and the problem is therefore trivial in cycles).

The distributed complexity of solving sinkless orientation is now completely understood: $\Theta(\log n)$ rounds with deterministic algorithms and $\Theta(\log\log n)$ rounds with randomized algorithms \citep{Brandt2016,chang16exponential,ghaffari17distributed}. Yet there are many fundamental questions related to the role of sinkless orientation that we do not understand at all.

One of the mysteries is the following observation: whenever we encounter a problem $\Pi$ that turns out to be as hard as sinkless orientation, usually the reason for $\Pi$ being hard is that $\Pi$ is directly related to sinkless orientation through reductions. For example, sinkless and sourceless orientation has the same complexity as sinkless orientation; the upper bound is nontrivial \cite{ghaffari17distributed}, but the lower bound trivially comes from the observation that a sinkless and sourceless orientation gives a sinkless orientation. We are not aware of any LCL problem $\Pi$ that has the same complexity as sinkless orientation---$\Theta(\log n)$ rounds deterministic and $\Theta(\log\log n)$ rounds randomized---with a nontrivial lower bound that is not merely an observation that an algorithm for solving $\Pi$ directly gives an algorithm for solving sinkless orientation.

Indeed, we do not know if all problems in this complexity class are merely extensions and variants of the same problem!

One of the goals of the present work is to systematically explore a large family of LCL problems that contains the sinkless orientation problem as well as many closely related problems. Our definition of binary labeling problems is guided by this insight: it is large enough to contain many such problems, yet restrictive enough so that there is hope of understanding the distributed complexity of all problems in this family.

As we will see in this work, the study of binary labeling problems sheds new light on this issue. We will encounter problems that are as hard as sinkless orientation, yet require a \emph{new lower bound proof} that is not based on the previous result of the hardness of sinkless orientation.

\paragraph{Binary labeling problems vs.\ homogeneity.}

LCLs come in many flavors. Some problems are trivial in a regular unlabeled graph---fractional problems are a good example here. However, for many distributed problems the interesting part is specifically what to do in the middle of a regular unlabeled tree---symmetry-breaking problems and orientation and splitting problems fall in this category. Our focus is on the latter case.

In the definition of binary labeling problems, the idea of focusing on regular parts is simply captured in the constraint that only white nodes of degree $d$ and only black nodes of degree $\delta$ are relevant. In essence, as soon as we see some irregularities (e.g., leaf nodes of the tree), the problem becomes potentially easier to solve.

The idea of investigating what happens in the middle of a regular tree was previously formalized using a somewhat different idea of \emph{homogeneous LCLs} \cite{Balliu2019weak2col}. The family of problems that we study in the present work is similar in spirit to homogeneous LCLs, but there is a key difference: all homogeneous LCLs can be solved in $O(\log n)$ rounds in trees, but there are binary labeling problems with distributed time complexity $\Theta(n)$; see Table~\ref{tab:overview}.

\begin{table}
\newcommand{\yes}{{\footnotesize\textcolor{green!50!black}{YES}}}
\newcommand{\no}{{\footnotesize\textcolor{red!50!black}{NO}}}
\newcommand{\hlno}{\textbf{\no}}
\newcommand{\hlyes}{\textbf{\yes}}
\newcommand{\open}{?}
\small
\centering
\begin{tabular}{llllll}
\toprule
Deterministic & Randomized & General & Homogeneous & \textbf{Binary} & Examples \\
complexity & complexity & LCLs & LCLs & \textbf{labeling} & \\
\midrule
$O(1)$ & $O(1)$ & \yes & \yes & \hlyes & trivial problems \\
$\omega(1)$, $o(\log^* n)$ & $\omega(1)$, $o(\log^* n)$ & \open & \no{} \cite{Balliu2019weak2col} & \hlno & \open \\
$\Theta(\log^* n)$ & $\Theta(\log^* n)$ & \yes & \yes & \hlno & $(\Delta+1)$-coloring \cite{Linial1992,cole86deterministic} \\
$\Theta(\log n)$ & $\Theta(\log\log n)$ & \yes & \yes & \hlyes & sinkless orientation \citep{Brandt2016,chang16exponential,ghaffari17distributed} \\
$\Theta(\log n)$ & $\Theta(\log n)$ & \yes & \yes & \hlyes & even orientation (Example~\ref{ex:even-orientation}) \\
$\Theta(n^{1/k})$ & $\Theta(n^{1/k})$ & \yes & \no{} \cite{Balliu2019weak2col} & \hlno & $2\frac12$-coloring \cite{Chang2019} \\
$\Theta(n)$ & $\Theta(n)$ & \yes & \no{} \cite{Balliu2019weak2col} & \hlyes & 2-coloring (Example~\ref{ex:two-coloring}) \\
\bottomrule
\end{tabular}
\caption{An overview of possible distributed time complexities in trees for different classes of LCLs.}\label{tab:overview}
\end{table}

\paragraph{Automatic round elimination and fixed points.}

As we mentioned in Section~\ref{ssec:binary-labeling-problems}, the edge labeling formalism makes it easy to apply the automatic round elimination technique \cite{Brandt2019automatic,Olivetti2019}, which is an effective technique for proving lower bounds \cite{Balliu2019,Brandt2019automatic,Brandt2016,Balliu2019weak2col,Linial1992,Naor1991}. However, there are still many open research questions related to round elimination itself. One of the key questions is the existence of fixed points; here the sinkless orientation problem is a good example.

Let $\Pi$ be the sinkless orientation problem from Example~\ref{ex:sinkless-orientation}. If we start with the assumption that the complexity of $\Pi$ is $T = o(\log n)$ rounds in regular trees, and apply the round elimination technique \cite{Brandt2019automatic,Olivetti2019} in a mechanical manner, we will immediately get the result that the complexity of the same problem $\Pi$ is $T-1$ rounds, which is absurd, and hence we immediately get a lower bound (at least for some weak models of distributed computing). A bit more formally, $\Pi$ is a nontrivial fixed point for round elimination, and such fixed points immediately imply nontrivial lower bounds. This is particularly interesting as fixed points can be detected automatically with a computer.

Now let $\Pi'$ be the problem of finding a sinkless and sourceless orientation from Example~\ref{ex:sinkless-sourceless}. For a human being, it is now trivial that $\Pi'$ is at least as hard as $\Pi$. However, $\Pi'$ is not a fixed point for round elimination, and we do not seem to have any automatic way for proving lower bounds for problems similar to $\Pi'$.

What is not understood at all is what is the fundamental difference between $\Pi$ and $\Pi'$ and why one of them is a fixed point and the other one is not. Indeed, so far there have not been many examples of nontrivial fixed points that are not merely trivial variants of sinkless orientations.

The family of binary labeling problems is chosen so that it contains both fixed points as well as closely related problems that are not fixed points. As we will see in this work, through a systematic study of all binary labeling problems we are able to discover new nontrivial fixed points.

\paragraph{\boldmath Restriction to two labels and $\Theta(\log^* n)$ complexity.}

In the binary labeling problem we label edges (or half-edges) with two labels: whether it is part of solution $X$ or not. If we looked at the more general case of $|\Sigma| = O(1)$ possible edge labels, we would have a family of problems that is, in essence, more expressive than the family of all homogeneous LCLs.

The case of $|\Sigma| = 2$ that we study here is the smallest nontrivial case, but there is another reason that makes the case of $|\Sigma| = 2$ interesting: we will show that the \emph{complexity class $\Theta(\log^* n)$ disappears} when we go from $|\Sigma| = 3$ down to $|\Sigma| = 2$.

There are numerous symmetry-breaking problems that can be solved in $\Theta(\log^* n)$ rounds in graphs of maximum degree $\Delta=O(1)$. Examples include maximal matching, maximal independent set, minimal dominating set, $(\Delta+1)$-vertex coloring, $(2\Delta-1)$-edge coloring, weak $2$-coloring, and many variants of these problems.

It is known that with $3$ edge labels we can encode e.g.\ the problem of finding a maximal matching~\cite{Balliu2019}; hence as soon as $|\Sigma| \ge 3$, there are edge labeling problems solvable in $\Theta(\log^* n)$ rounds. However, previously it was not known if the use of $|\Sigma| \ge 3$ labels was necessary in order to encode any such problem. After all, the use of $|\Sigma| \ge 3$ labels seems unnatural when we consider problems such as maximal matching, maximal independent set, and weak $2$-coloring, all of which in essence ask one to find a subset of edges or a subset of nodes subject to local constraints.

In this work we will show that none of these problems can be encoded with two labels. We show that for binary labeling problems, there is a gap in the deterministic complexity between $\omega(1)$ and $o(\log n)$.

Hence binary labeling problems give rise to a different landscape of computational complexity in comparison with any other number of labels. We conjecture that there is no such qualitative difference between e.g.\ alphabets of size $3$ or~$4$.

\section{Our contributions}\label{sec:contributions}

\subsection{Main results: deterministic complexity}

\begin{table}
\newcommand{\unsol}{unsolvable}
\small
\centering
\begin{tabular}{llll@{}l@{\qquad}lll}
\toprule
Type & Problem & White & Black && \textbf{Deterministic} & Upper & Lower \\
& family & constraint & constraint && \textbf{complexity} & bound & bound \\
\midrule
$\tI$&$\fIa$& $\p\n\nnn$ & $\n\x\xxx$ && \textbf{\unsol} && Section~\ref{sec:unsol} \\
&$\fIb$& $\n\nnn\p$ & $\x\xxx\n$ \\
&$\fIc$& $\n\x\xxx$ & $\p\n\nnn$ \\
&$\fId$& $\x\xxx\n$ & $\n\nnn\p$ \\\addlinespace
$\tII$&$\fIIa$& $\n\n\nnn$ & $\x\x\xxx$ \\
&$\fIIb$& $\x\x\xxx$ & $\n\n\nnn$ \\
\midrule
$\tIII$&$\fIIIa$& non-empty & $\p\p\ppp$ && \boldmath $O(1)$ & Section~\ref{sec:constant} \\
&$\fIIIb$& $\p\p\ppp$ & non-empty \\\addlinespace
$\tIV$&$\fIVa$& $\p\x\xxx$ & $\p\x\xxx$ \\
&$\fIVb$& $\x\xxx\p$ & $\x\xxx\p$ \\
\midrule
$\tV$&$\fVa$& $\p\nnn\p$ & $\n\p\n$ && \boldmath $\Theta(n)$ & Section~\ref{sec:global} & Section~\ref{sec:global} \\
&$\fVb$& $\n\p\n$ & $\p\nnn\p$ \\\addlinespace
$\tVI$&$\fVIa$& $\nnn\p\x$ & $\x\p\nnn$ \\
&$\fVIb$& $\x\p\nnn$ & $\nnn\p\x$ \\
\midrule
$\tVII$&$\fVIIa$& \multicolumn{2}{l}{all other cases} && \boldmath $\Theta(\log n)$ & Section~\ref{sec:log-upper} & Section~\ref{sec:log-lower} \\
\bottomrule
\end{tabular}
\caption{The deterministic distributed complexity of binary labeling problems in trees; $\ppp$ signifies one or more $\p$s, and $\x$ signifies either $\n$ or $\p$. Note that e.g.\ all problem families of type $\tI$ are equivalent to each other in the following sense: for any problem $\Pi$ of type $\tI$, there is exactly one problem equivalent to $\Pi$ in each of the four families $\fIa$, $\fIb$, $\fIc$, and $\fId$.}\label{tab:deterministic}
\end{table}

Our main contribution is a complete classification of the deterministic distributed complexity of binary labeling problems in trees---this is presented in Table~\ref{tab:deterministic}. Given a binary labeling problem $\Pi$, one can simply do pattern matching in this table to first find its problem family and this way also find its deterministic complexity.

For easier reference, we have listed all problem families explicitly, but if we take into account the fact that e.g.\ all problems of family $\fIa$ are equivalent to a problem in family $\fIb$ and vice versa (recall Observation~\ref{obs:equiv}), we can partition problems in seven \emph{types}, labeled with $\tI, \tII, \dotsc, \tVII$ in the table.

All proofs and algorithms are presented in Sections \ref{sec:unsol}--\ref{sec:log-lower}: Section~\ref{sec:unsol} covers the complexity of problems of types $\tI$ and $\tII$, Section~\ref{sec:constant} covers types $\tIII$ and $\tIV$, and Section~\ref{sec:global} covers types $\tV$ an $\tVI$. Now by definition all problems that are not of type $\tI$--$\tVI$ are of type $\tVII$; in Section~\ref{sec:log-upper} we prove that all of them are solvable in $O(\log n)$ rounds and in Section~\ref{sec:log-lower} we prove a matching lower bound of $\Omega(\log n)$ for all of these problems. This will then complete the proof that the classification of Table~\ref{tab:deterministic} is correct.

The first two parts of Table~\ref{tab:deterministic} are trivial:
\begin{itemize}
    \item There are two types of problems that are unsolvable for trivial reasons. For example, if all white nodes must have $X$-degree $0$ and all black nodes must have a non-zero $X$-degree, no solution exists as long as the tree is large enough.
    \item There are two types of problems that are solvable in $O(1)$ time for trivial reasons. For example, if black nodes are happy with any $X$-degree, then white nodes can simply pick some number $x \in W$ and choose arbitrarily $x$ adjacent edges.
\end{itemize}
However, what is not obvious is that this list is exhaustive: no other problems are unsolvable, and no other problem is solvable in $O(1)$ time.

Furthermore, as we can see in Table~\ref{tab:deterministic}, there are only very few binary labeling problems that are solvable but inherently global, i.e., they require $\Theta(n)$ rounds to solve. What is perhaps the biggest surprise is that all other problems can be solved in $O(\log n)$ rounds, and they also require $\Omega(\log n)$ rounds. In particular:
\begin{theorem}
There are no binary labeling problems with deterministic distributed complexity in the following ranges:
\begin{itemize}[noitemsep]
    \item between $\omega(1)$ and $o(\log n)$,
    \item between $\omega(\log n)$ and $o(n)$.
\end{itemize}
\end{theorem}

Table~\ref{tab:examples} shows examples of problems in each complexity class.

\begin{table}
\newcommand{\unsol}{unsolvable}
\small
\centering
\begin{tabular}{llllll}
\toprule
Deterministic & Type & \multicolumn{3}{l}{Examples of problems} \\
complexity && \\
\cmidrule{3-6}
&& $W$ & $B$ & reference & description\\
\midrule
\unsol
& $\tI$ & $\n\p\ppp$ & $\p\n\n$ & Example~\ref{ex:contradiction} & contradiction \\
\midrule
$O(1)$
& $\tIII$ & $\n\n\p\nnn$ & $\p\p\p$ & Example~\ref{ex:trivial} & trivial \\
\midrule
$\Theta(n)$
& $\tV$ & $\p\nnn\p$ & $\n\p\n$ & Example~\ref{ex:two-coloring} & two-coloring \\
\midrule
$\Theta(\log n)$
& $\tVII$ & $\n\p\p\p\n$ & $\n\p\p\p\n$ & Example~\ref{ex:bipartite-splitting} & bipartite splitting \\
&& $\p\n\p\n$ & $\n\p\n$ & Example~\ref{ex:even-orientation} & even orientation \\\addlinespace
&& $\p\p\ppp\n$ & $\n\p\n$ & Example~\ref{ex:sinkless-orientation} & sinkless orientation \\
&& $\n\p\ppp\n$ & $\n\p\n$ & Example~\ref{ex:sinkless-sourceless} & sinkless and sourceless orientation \\
&& $\n\p\n\nnn$ & $\p\n\p$ & Example~\ref{ex:regular-matching} & regular matching \\
&& $\n\p\ppp\n$ & $\p\n\p$ & Example~\ref{ex:splitting} & splitting \\\addlinespace
&& $\n\p\n\nnn$ & $\n\p\n\nnn$ & Definition~\ref{def:bipartite-matching} & bipartite matching \\
&& $\n\p\n\nnn$ & $\p\nnn\p$ & Definition~\ref{def:hypergraph-matching} & hypergraph matching \\
&& $\n\p\n\nnn$ & $\p\p\n$ & Definition~\ref{def:edge-grabbing} & edge grabbing \\
\midrule
$\Omega(\log n)$, $O(n)$
& $\tV$ or $\tVII$ & $\ppp\n\ppp$ & $\n\ppp\n$ & Definition~\ref{def:forbidden-degree} & forbidden degree \\
& $\tVI$ or $\tVII$ & $\p\ppp\n$ & $\n\p\ppp$ & Definition~\ref{def:bipartite-so} & bipartite sinkless orientation \\
\bottomrule
\end{tabular}
\caption{Examples of binary labeling problems and problem families and their deterministic complexities.}\label{tab:examples}
\end{table}

\subsection{Additional results: randomized complexity}

While our focus in this work is on deterministic complexity, we will also explore the randomized complexity of binary labeling problems. Many of our theorems from Sections \ref{sec:unsol}--\ref{sec:log-lower} have also direct implications on randomized complexity. We will summarize all results related to randomized complexity in Section~\ref{sec:rand}, and have a deeper look at questions that are specific to randomized complexity there.

By prior work, it is known that classes $O(1)$ and $\Theta(n)$ remain the same also for randomized complexity---this follows from \cite{Ghaffari2018} and (unpublished) extensions of the result in \cite{Balliu2018disc}. However, class $\Theta(\log n)$ is more interesting, as there are binary labeling problems of the following types:
\begin{itemize}
    \item Deterministic complexity $\Theta(\log n)$ and randomized complexity $\Theta(\log \log n)$ \\
    (e.g.\ sinkless orientation, Example~\ref{ex:sinkless-orientation}).
    \item Deterministic complexity $\Theta(\log n)$ and randomized complexity $\Theta(\log n)$ \\
    (e.g.\ even orientation, Example~\ref{ex:even-orientation}).
\end{itemize}
That is, randomness helps with some binary labeling problems but not all. While we present a partial classification of the randomized complexity of binary labeling problems, the main open question for the future work is coming up with a complete characterization of exactly which binary labeling problems can be solved in $O(\log \log n)$ rounds with randomized algorithms.

\section{Model of computing}\label{sec:model}

There are two equivalent perspectives that we can use to reason about locality in distributed algorithms:
\begin{enumerate}
    \item Each node sees its radius-$T$ neighborhood in the input graph and uses this information to choose its own part of solution.
    \item The input graph is a computer network and initially each node is only aware of its own input. Nodes can exchange messages with each other for $T$ communication rounds and then they will stop and announce their own part of the solution.
\end{enumerate}
In this work we will use the second approach to formalize the model of computing. Our formalism is equivalent to the usual LOCAL model of distributed computing \cite{Peleg2000,Linial1992}, up to some additive constants in the number of communication rounds.

We will first define deterministic distributed algorithms in the port-numbering model, which is a rather restricted model of computing. Then we will add unique identifiers to arrive at the deterministic LOCAL model, and finally randomness to arrive at the randomized LOCAL model.

\paragraph{Port-numbering model.}

Given a graph $G = (V,E)$, define the set of \emph{ports} by
\[
	P = \bigl\{ (v,i) : v \in V, i \in \{1,2,\dotsc,\deg(v)\} \bigr\}.
\]
We say that $p \colon P \to P$ is a \emph{port numbering} of $G$ if for each edge $\{u,v\} \in E$ there are natural numbers $1 \le i \le \deg(u)$ and $1 \le j \le \deg(v)$ such that $p(u,i) = (v,j)$ and $p(v,j) = (u,i)$. We use e.g.\ the following terminology:
\begin{itemize}[noitemsep]
	\item Port $i$ of node $u$ is connected to port $j$ of node $v$.
	\item Port $i$ of node $u$ is connected to node $v$.
	\item Node $v$ is the $i$th neighbor of node $u$.
\end{itemize}
A \emph{port-numbered network} is a graph $G$ together with a port numbering $p$. We can interpret a port-numbered network as a computer network in which a node $v$ is a computer with $\deg(v)$ communication ports and e.g.\ $p(u,i) = (v,j)$ indicates that if node $u$ sends a message to communication port $i$, it is received by node $v$ from its communication port $j$, and vice versa.

Computation in a port-numbered network $(G,p)$ proceeds as follows:
\begin{itemize}
	\item All nodes \emph{initialize} their local states. The initial state can only depend on the \emph{local input} of the node and on its degree.
	\item Then we repeat \emph{synchronous communication rounds} until all nodes have stopped. In each round:
	\begin{enumerate}
		\item All nodes that have not stopped yet choose what messages to send to each of their communication ports.
		\item Messages are forwarded along the edges to the recipients.
		\item All nodes that have not stopped yet update their local states based on the messages that they received. Optionally, a node may choose to stop and announce its \emph{local output}.
	\end{enumerate}
\end{itemize}
A \emph{distributed algorithm} $\cA$ in this model is simply a collection of three procedures: one for initialization, one for constructing outgoing messages, and one for doing state updates after receiving messages.

\paragraph{White and black algorithms.}

We say that $\cA$ is a \emph{white algorithm} that solves an edge labeling problem $\Pi$ in time $T(n)$ if the following holds:
\begin{itemize}
	\item We can take any graph $G = (V,E)$ and any two-coloring $f\colon V \to \{ \black, \white \}$ of $G$.
	\item We can take any port numbering $p$ of $G$.
	\item If we run $\cA$ in port-numbered network $(G,p)$ with local inputs $f$, all nodes will stop after at most $T(|V|)$ communication rounds and announce their local outputs.
	\item For each white node $u$ the local output specifies the output label for each incident edge (as a vector of $\deg(u)$ output labels from some alphabet $\Sigma$, ordered by port numbers).
	\item For each black node $v$ the local output is empty.
	\item The local outputs of the white nodes constitute a feasible solution for $\Pi$, i.e., they satisfy all white constraints and all black constraints.
\end{itemize}
In particular, if $\Pi$ is a binary labeling problem $\Pi = (d,\delta,W,B)$, we require that the outputs of the white nodes can be interpreted as a subset of edges $X \subseteq E$ such that $X$ satisfies both white and black constraints.

We say that $\Pi$ has \emph{white complexity} $T$ in the port-numbering model if $T$ is the pointwise minimum of all functions $T'$ such that there exists a white algorithm $\cA$ that solves $\Pi$ in time $T'$.

We define a \emph{black algorithm} in analogous manner: black nodes produce output labels for their incident edges and white nodes produce empty outputs. \emph{Black complexity} is the minimum of the running times of all black algorithms.

\begin{observation}\label{obs:white-black-swap}
If we have $\Pi = (d,\delta,W,B)$ and $\Pi' = (\delta,d,B,W)$, then a white algorithm for $\Pi$ can be interpreted as a black algorithm for $\Pi'$ and vice versa.
\end{observation}

\begin{observation}\label{obs:white-black-algorithm}
Black complexity of problem $\Pi$ and white complexity of the same problem $\Pi$ differ by at most $1$ round.
\end{observation}
\begin{proof}
Given a white algorithm $\cA$ that solves $\Pi$ in time $T$, we can construct a black algorithm $\cA'$ that solves $\Pi$ in $T+1$ rounds: run $\cA$ for $T$ rounds, and then use $1$ round to have all white nodes inform their black neighbors about the final output. Conversely, given a white algorithm $\cA$ that solves $\Pi$ in time $T$, we can construct a white algorithm $\cA'$ that solves $\Pi$ in $T+1$ rounds.
\end{proof}

\paragraph{Deterministic LOCAL model.}

In the LOCAL model we augment the port-numbering model with unique identifiers. We assume that if there are $n$ nodes in the graph, each node is labeled with a unique integer from $\{1,2,\dotsc,\poly(n)\}$. This is part of the local input, and a distributed algorithm can use it when it initializes the local state of a node. Other than that, computation proceeds exactly in the same way as in the port-numbering model.

We say that $\cA$ is a deterministic white algorithm in the LOCAL model for problem $\Pi$ if it solves $\Pi$ correctly for any assignment of unique identifiers. We define black algorithms, deterministic white complexity in the LOCAL model, and deterministic black complexity in the LOCAL model in the analogous way.

For brevity, when we refer to a (deterministic) algorithm or (deterministic) complexity, we mean deterministic white algorithms in the LOCAL model.

\paragraph{Randomized algorithms.}

Finally, we can let state transitions be probabilistic, and we arrive at randomized algorithms. We can define e.g.\ randomized white algorithms in the LOCAL model in the natural manner. In this work, randomized algorithms are Monte Carlo algorithms that are correct w.h.p.\ in the size of the graph.

For brevity, when we refer to a randomized algorithm or randomized complexity, we mean randomized white algorithms in the LOCAL model.

\paragraph{Simulations.}

As discussed in Section~\ref{ssec:case-delta2}, it is often convenient to e.g.\ interpret a graph $G_0$ as a bipartite graph $G$ by subdividing edges. Now if we have an algorithm $\cA$ that we can apply in $G$, we can easily simulate $\cA$ in $G_0$.

There are some technical details that we need to take into account. A port numbering of $G$ induces a port numbering of $G_0$. However, a port numbering of $G_0$ does not directly give port numbers for the black nodes of $G$. However, we are usually interested in algorithms for the LOCAL model of computing, and we can use the unique identifiers to choose the port numbers for all nodes of $G$ (e.g., port $1$ points towards the neighbor with the smallest identifier).

Each node of the (virtual) network $G$ has to be simulated by some node of the (physical) network $G_0$. In particular, each black node has to be simulated by one of its white neighbors. We can use the port numbers or unique identifiers to choose who simulates what.

Finally, for the purposes of lower bounds we can also go in the other direction and take an algorithm for $G_0$ and simulate it in graph $G$. This increases the round complexity by a factor of two, which we can ignore as we are only interested in asymptotics.

\section{Unsolvable problems}\label{sec:unsol}

In this section we prove that any binary labeling problem contained in the following families is unsolvable:
\begin{itemize}[noitemsep]
\newcommand{\fIbox}[1]{\makebox[\widthof{$\fIa$:}]{#1}\,}
\newcommand{\fIIbox}[1]{\makebox[\widthof{$\fIIa$:}]{#1}\,}
\item \fIbox{$\fIa$:} $W = \p\n\nnn$, $B = \n\x\xxx$,
\item \fIbox{$\fIb$:} $W = \n\nnn\p$, $B = \x\xxx\n$,
\item \fIbox{$\fIc$:} $W = \n\x\xxx$, $B = \p\n\nnn$,
\item \fIbox{$\fId$:} $W = \x\xxx\n$, $B = \n\nnn\p$,
\medskip
\item \fIIbox{$\fIIa$:} $W = \n\n\nnn$, $B = \x\x\xxx$,
\item \fIIbox{$\fIIb$:} $W = \x\x\xxx$, $B = \n\n\nnn$.
\end{itemize}
This completes the classification in the first section of Table~\ref{tab:deterministic}, i.e., the classification of type \tI{} and \tII. We refer to these problems as \emph{unsolvable} problems but we emphasize that this section does not contain a proof that there are no other unsolvable problems. We will only learn this later as Sections \ref{sec:constant}--\ref{sec:log-upper} will show that all problems of type \tIII--\tVII{} are solvable.

\begin{theorem}\label{UnsolvableProblemsthm}
All binary LCL problems that are contained in any of the above families cannot be solved on any  tree meeting the degree requirements of the respective problem. 
\end{theorem}
\begin{proof}
The  proof will consider two cases. 
\begin{itemize}
    \item \textbf{\boldmath Families  $\fIa,\fIb,\fIc,\fId$:}   Fix problem $\Pi=(d,\delta,W,B)\in\fIc$ with $W=\n\p^{d}$, $B=\p\n^{\delta}$. Then $\Pi$ is \emph{easier} than any problem $\Pi'$ in the family $\fIc$, i.e., $\Pi'\subseteq \Pi$; furthermore any problem in family $\fIa$, $\fIb$ or $\fId$ is equivalent (as defined in Subsection \ref{sssec:equivalence}) to some problem in $\fIc$; thus $\Pi$ is the \emph{easiest} problem in the considered families and if $\Pi$ is unsolvable any problem in all the families is unsolvable.
		
We now show that $\Pi$ is unsolvable. Given some tree $G=(V,E)$ that meets the degree requirements of $\Pi$, assume there is a valid solution to $\Pi$ in $G$ and consider some white node $v\in V$. Since white nodes are not allowed to have their incident edges labeled all $0$, there must be an incident edge labeled $1$, that we denote by $e=\{v,u\}$. However, this means that the black node $u$ has the incident edge $e$ labeled with $1$. This contradicts our assumption of having a valid solution to $\Pi$, since black nodes are only allowed to label all of their incident edges with $0$. Thus $\Pi$ is unsolvable and the proof is complete.
    
    \item \textbf{\boldmath Families $\fIIa,\fIIb$:} Fix problem $\Pi=(d,\delta,W,B)\in\fIIb$ with $W=\p^{d+1}$, $B=\n^{\delta+1}$. Similar to the reasoning in the previous case it is sufficient to show that $\Pi$ is unsolvable to deduce that any problem in the two families is unsolvable. 
		 Given some tree $G=(V,E)$, assume that $\Pi$ can be solved on $G$ and consider some black node $v\in V$. Denote by $i$ the number of edges incident to $v$ that are labeled with $1$. Clearly we have $0\leq i\leq d$. However, no matter the value of $i$, black nodes are not allowed $i$ incident edges labeled with $1$ in their solution, thus we arrive at a contradiction. Thus $\Pi$ is unsolvable and the proof is complete. 
\qedhere
\end{itemize}
\end{proof}

\section{Constant-time problems}\label{sec:constant}

In this section we will prove that any binary labeling problem in the following families can be solved in constant time:
\begin{itemize}[noitemsep]
\newcommand{\fIIIbox}[1]{\makebox[\widthof{$\fIIIa$:}]{#1}\,}
\newcommand{\fIVbox}[1]{\makebox[\widthof{$\fIVa$:}]{#1}\,}
\item \fIIIbox{$\fIIIa$:} $W \ne \n\n\nnn$, $B = \p\p\ppp$,
\item \fIIIbox{$\fIIIb$:} $W = \p\p\ppp$, $B \ne \n\n\nnn$,
\medskip
\item \fIVbox{$\fIVa$:} $W = \p\x\xxx$, $B = \p\x\xxx$,
\item \fIVbox{$\fIVb$:} $W = \x\xxx\p$, $B = \x\xxx\p$.
\end{itemize}
This completes the classification in the second section of Table~\ref{tab:deterministic}, i.e., the classification of types \tIII{} and \tIV.  We refer to these problems as \emph{trivial} problems but we emphasize that this section does not contain a proof that there are no other trivial problems, i.e., that there are not other problems solvable in time $O(1)$. We will only learn this later as Section \ref{sec:unsol} showed that problems of type \tI{} and \tII{} are unsolvable and Section \ref{sec:global} and Section \ref{sec:log-lower} will show that all problems of type \tV--\tVII{} have $\omega(1)$ lower bounds.

\begin{theorem}
All binary LCL problems that are contained in any of the above families can be solved in constant time on any  tree meeting the degree requirements of the respective problem. 
\end{theorem}
\begin{proof}
The proof considers two cases.
\begin{itemize}
    \item \textbf{\boldmath Families $\fIIIa, \fIIIb$:} Any problem $\Pi'\in\fIIIb$ has an equivalent problem in $\fIIIa$. Thus it is sufficient to show that any problem in $\fIIIa$ can be solved in constant time.
    Let $\Pi\in\fIIIa$ be a problem with white degree $d$. Let $G=(V,E)$ be some input tree. We know that for $\Pi$, it holds that $W\neq \n^{d}$. Thus there exists $0\leq i\leq d$ such that $w_i=\p$. Thus $\Pi$ allows white nodes to label $i$ of their incident edges with $1$. 
    Now, each white node arbitrarily picks $i$ incident edges and labels them with $1$ and labels all other of its incident edges with $0$. Thus the labeling of the incident edges around any white node is according to what $W$ allows. The labeling is also feasible for every black node as $B = \p\p\ppp$ says that any labeling is allowed for black nodes. Clearly this takes constant time. 		
		
    \item \textbf{\boldmath Families $\fIVa, \fIVb$:} Any problem in $\fIVa$ can be solved by labeling all edges with $0$ and any problem in $\fIVb$ can be solved by labeling all edges with $1$. Clearly these are constant time algorithms. 
\qedhere
\end{itemize}
\end{proof}

\section{Global problems}\label{sec:global}

In this section we will prove that any binary labeling problem in the following families can be solved in $O(n)$ rounds and they also require $\Omega(n)$ rounds:
\begin{itemize}[noitemsep]
\newcommand{\fVbox}[1]{\makebox[\widthof{$\fVa$:}]{#1}\,}
\newcommand{\fVIbox}[1]{\makebox[\widthof{$\fVIa$:}]{#1}\,}
\item \fVbox{$\fVa$:} $W = \p\nnn\p$, $B = \n\p\n$,
\item \fVbox{$\fVb$:} $W = \n\p\n$, $B = \p\nnn\p$,
\medskip
\item \fVIbox{$\fVIa$:} $W = \nnn\p\x$, $B = \x\p\nnn$,
\item \fVIbox{$\fVIb$:} $W = \x\p\nnn$, $B = \nnn\p\x$.
\end{itemize}
This completes the classification in the third section of Table~\ref{tab:deterministic}, i.e., the classification of types \tV{} and \tVI. We refer to these problems as \emph{global} problems but we emphasize that this section does not contain a proof that there are no other global problems. We only learn this because previous sections showed that problems of type \tI--\tIV{} are either unsolvable or trivial and Section \ref{sec:log-upper} will show that problems of \tVII{} can be solved in $O(\log n)$ rounds.

We begin with showing that $\fVa$ and $\fVb$ have complexity $\Theta(n)$
\begin{theorem}
   $\fVa$ and $\fVb$ have complexity $\Theta(n)$.
   \end{theorem}
   \begin{proof}
  $\fVb$ is equivalent to $\fVa$ thus we only need to show that $\fVa$ has complexity $\Theta(n)$ to prove the theorem. Recall that $\fVa=(d,2,W,B)$ with $W=\p\nnn\p$ and $B=\n\p\n$. As explained in Section~\ref{ssec:case-delta2}, since the black degree $\delta$ equals $2$, we can treat the problem as an orientation problem on a given $d$-regular tree $G=(V,E)$ in which all nodes are white. All nodes have to satisfy the constraint string $W=\p\nnn\p$, thus, for all nodes $v$, the resulting orientation of the edges must satisfy either $\indegree(v)=0$ or $\indegree(v)=d$, where $\indegree(v)$ is the number of edges incident to $v$ which are oriented towards $v$. Now consider two adjacent nodes of degree $d$ in $G$, denoted by $v,u$. A key observation is that in any valid orientation of the edges, it holds that $\indegree(v)=0$ implies $\indegree(u)=d$ and vice versa. Now, if we treat $\indegree(v)=d$ and $\indegree(v)=0$ as two colors, we can deduce that the subgraph of $G$ induced by nodes of degree $d$ has to be properly colored with $2$ colors in any valid solution to $\fVa$. 
      We use this observation to show that the $\fVa$ has complexity $\Omega(n)$. 
     Consider the tree obtained by starting from a path $P$ and connecting to each node $d-2$ new nodes. Also, connect an additional node to each endpoint of the path to make the tree $d$-regular. In a tree of $n$ nodes, this path $P$ has $\Omega(n)$ nodes. Note that each node in the path has degree $d$ and thus must satisfy the constraint of $W$.  Thus, in any valid solution for $\fVa$, the nodes of $P$ must be properly $2$-colored.     
      Thus any algorithm $A$ solving $\fVa$ in $T$ rounds would be able to solve the problem of $2$-coloring a path properly in $O(T)$ rounds by simply letting each node in a given path simulate $A$ with some additional virtual nodes connected to it as leaves. It is well known that any algorithm solving 2-coloring on a path has complexity $\Omega(n)$. Thus we can deduce that any algorithm solving $\fVa$ has complexity $\Omega(n)$, as required. 
      
      To show that $\fVa$ can also be solved in $O(n)$ rounds it is sufficient that $\fVa$ is solvable on any $d$-regular tree. Any tree is $2$-colorable and orienting all edges from one color family to the other solves the problem. 
   \end{proof}   

As any problem in $\fVIb$ is equivalent to a problem in $\fVIa$ we will only show that the easiest problem in $\fVIb$ has complexity $\Omega(n)$ to prove the following theorem. 
\begin{theorem}\label{rootedtreeglobal}
Any problem in family $\fVIa$ or $\fVIb$ has complexity $\Omega(n)$.
\end{theorem}
\begin{proof}
Any problem in $\fVIa$ is equivalent to a problem in $\fVIb$. Thus it is sufficient to show that the \emph{easiest} problem in $\fVIb$ has complexity $\Omega(n)$ to show that all problems in $\fVIb$ and $\fVIa$ have complexity $\Omega(n)$. 
Let $\Pi=(d,\delta,W,B)\in\fVIb$ where  $W = \p\p\nnn$, $B = \nnn\p\p$. Problem $\Pi$ is the easiest problem in $\fVIb$ as $\Pi'\subseteq \Pi$ for any problem $\Pi'\in\fVIb$. We show that $\Pi$ has complexity $\Omega(n)$.

  Let $X \subseteq E$ be the set of marked edges. The set $X$ defines an orientation of the edges, given by orienting edges in $X$ from black to white nodes, and edges in $E\setminus X$ from white to black nodes. Similarly, an orientation defines an assignment for $\Pi$. It is easy to see that $\Pi$ requires to orient edges such that nodes have either 0 or 1 incoming edges. 
  
  We now prove that there can be at most a single vertex in the graph having all edges oriented outgoing, implying that $\Pi$ is equivalent to orienting a tree from some root, namely, to the \emph{oriented tree} problem. Assume towards a contradiction that there are two nodes $u,v$ having all outgoing edges. Denote by $P$ the unique path between $u,v$ in the tree. Starting from $u$, denote by $w$ the first vertex on $P$ such that the edges incident to $w$ that are part of $P$ are not both oriented towards $v$. Such a vertex must exist, because the incident edge of $v$ on $P$ is not directed towards $v$, as all its edges are oriented outgoing. Since $w$ is the first node satisfying the condition, both the incident edges of $w$ on $P$ are oriented towards $w$, thus $w$ has two incoming edges, giving a contradiction.

   All that is left to do, is to prove that the \emph{oriented tree} problem is indeed global. We reduce another problem, namely the \emph{almost oriented path} problem, to the oriented tree problem, and then we show that the almost oriented path problem is hard. The almost oriented path problem requires to orient the edges of a properly $2$-colored path such that at most one node of degree $2$ has all outgoing edges, and all other nodes of degree $2$ have exactly one outgoing edge (nodes of degree $1$ are unconstrained). 
   
   Now, let us show that an algorithm for the oriented tree problem can be used to solve the almost oriented path problem. In particular, we show how to locally turn an instance of one problem to the other. Start from a properly $2$-colored path, and connect to each white node $d-2$ new black nodes, and to each black node $\delta-2$ new white nodes. Then, connect an additional node to each endpoint of the path, preserving a proper $2$-coloring. Clearly, a solution for the oriented tree problem for this instance can be directly mapped to a solution for the almost oriented path problem on the original instance.
   
   We now prove that the almost oriented path problem requires $\Omega(n)$ time. Assume by contradiction that there is a $T=o(n)$ algorithm $A$ that solves the problem, and consider an instance $P$ of size $n$ such that $T(n) < n/100$. Let $u$, $v$ be the endpoints of the path. Depending on the output of $A$, we can give \emph{types} to nodes. A node is of type $L$ if both its incident edges are oriented towards $u$. A node is of type $R$ if both its incident edges are oriented towards $v$. Otherwise (if both edges are outgoing), it is of type $C$. 
   Consider four nodes $v_1,v_2,v_3,v_4$ at distance at least $3T$ between each other and from the endpoints of the path, where $v_1$ is the nearest to $u$ among them, and $v_4$ is the nearest to $v$. Notice that at least one of the two following statements must hold: both $v_1$ and $v_2$ are of type $L$, or both $v_3$ and $v_4$ are of type $R$, since otherwise there must be a node in the path having two incoming edges. Let assume w.l.o.g.\ that we are in the first case. Let $Q$ be the subpath of length $2T+1$ centered at node $v_1$.  We now construct a new path $P'$, by starting from $P$ and replacing $Q$ by its mirror. The algorithm $A$, while running on $v_1$ and $v_2$, cannot distinguish the instances $P$ and $P'$, thus it must output the same on both instances, but since $Q$ has been mirrored, this implies that $v_1$ on $P'$ has type $R$. Since in $P'$ the node $v_1$ has type $R$ and $v_2$ has type $L$, this implies that there is some node in the subpath between $v_1$ and $v_2$ having two incoming edges, contradicting the correctness of $A$.
   \end{proof}
\begin{theorem}
Any problem in family $\fVIa$ or $\fVIb$ has complexity $O(n)$.
\end{theorem}
\begin{proof}
Let $\Pi=(d,\delta,W,B)$ where  $W = \n\p\nnn$, $B = \nnn\p\n$. 
Any problem in $\fVIa$ is equivalent to a problem in $\fVIb$ and $\Pi\in\fVIb$ is the \emph{hardest} problem in $\fVIb$, i.e., $\Pi\subseteq \Pi'$ for all $\Pi'\in\fVIb$. Thus it is sufficient to show that $\Pi$ can be solved in $O(n)$ rounds. 

  Since any problem that is solvable can be solved in $O(n)$ rounds, it is enough to show that a solution for $\Pi$ always exists in any tree that meets the degree requirements of $\Pi$. It is clear that the problem described in Theorem \ref{rootedtreeglobal} is always solvable, since it is always possible to orient a tree towards a root. The current problem $\Pi$ is slightly harder, since it does not allow any white node of degree $d$ and any black node of degree $\delta$ to have all incoming edges. We can solve this problem anyway, by choosing a leaf of the tree as a root, i.e., a node with degree $1$ and orient all edges towards the leaf. 
\end{proof}

\section{Logarithmic upper bounds}\label{sec:log-upper}

In this section we will show that all binary labeling problems that are not of the type discussed in Section~\ref{sec:unsol}, Section~\ref{sec:constant}, or Section~\ref{sec:global} can be solved in $O(\log n)$ rounds with deterministic distributed algorithms. This will give the upper bounds for the final section of Table~\ref{tab:deterministic}. In this section, we formally prove the following theorem.

\begin{theorem}\label{thm:log-upper}
Any problem that is not of type \tI, \tII, \tIII, \tIV, \tV{} or \tVI{} has complexity $O(\log n)$.
\end{theorem}

To prove Theorem \ref{thm:log-upper}, we first introduce the concept of \emph{resilience}, and three families of binary labeling problems: \emph{bipartite matching}, \emph{hypergraph matching}, and \emph{edge grabbing}. In Lemma \ref{calim2section6} we will show that all problems that we have not yet covered in previous sections are resilient or they are contained in one of the three families.  Later, in Section~\ref{ssec:-resilupperbound}, we will show that resilient problems are solvable in $O(\log n)$ rounds, and finally in Sections \ref{ssec:bipartite-matching}--\ref{ssec:edge-grabbing} we show that all problems in the three families can also be solved in $O(\log n)$ rounds.

\subsection{Definitions}\label{ssec:log-upper-def}

\begin{definition}[Resilience]\label{resiliencedef}
Let $\Pi=(d,\delta,W,B)$ be a binary labeling problem. For $0\leq t \leq d+1$ and $0\leq s\leq \delta+1$ we say that $\Pi$ is \emph{$(t,s)$-resilient} if both of the following hold:
\begin{itemize}[noitemsep]
    \item bit string $W$ does not contain a substring of the form $\n^{d+1-t}$,
    \item bit string $B$ does not contain a substring of the form $\n^{\delta+1-s}$.
\end{itemize}
\end{definition}
\begin{remark}
Recall the \emph{even orientation} problem from Example~\ref{ex:even-orientation}: $W = \p\n\p\n$ and $B = \n\p\n$. This is a $(2,1)$-resilient problem: $W$ does not contain a substring $\n\n$, and $B$ does not contain a substring~$\n\n$.
\end{remark}

\begin{definition}[Bipartite matching]\label{def:bipartite-matching}
Problems of the form $\Pi = (\n\p\n\nnn, \n\p\n\nnn)$ are called \emph{bipartite matching problems}.
\end{definition}
\begin{remark}
Bipartite matching problems satisfy $d \ge 3$ and $\delta \ge 3$. In a $d$-regular bipartite graph, any solution to the bipartite matching problem is a perfect matching.
\end{remark}

\begin{definition}[Hypergraph matching]\label{def:hypergraph-matching}
Problems of the form $\Pi = (\n\p\n\nnn, \p\nnn\p)$ are called \emph{hypergraph matching problems}.
\end{definition}
\begin{remark}
Hypergraph matching problems satisfy $d\ge 3$ and $\delta \ge 2$. By interpreting black nodes as hyperedges the problem equals the hypergraph matching problem with hyperedges of rank $\delta$. In any $d$-regular hypergraph any such matching is a perfect hypergraph matching.
\end{remark}

\begin{definition}[Edge grabbing]\label{def:edge-grabbing}
Problems of the form $\Pi = (\n\p\n\nnn, \p\p\n)$ are called \emph{edge grabbing problems}.
\end{definition}
\begin{remark}
Edge grabbing problems satisfy $d\geq 3$ and $\delta=2$. If we interpret black nodes as edges of a simple graph, in a solution to the edge grabbing problem each white vertex \emph{grabs} exactly one of its incident edges where no edge is grabbed by both of its endpoints.
\end{remark}

\subsection{Coverage}

We now show that the definition of resilience and the three problem families defined above cover all binary labeling problems that are not unsolvable, trivial, or global. Recall Observation \ref{obs:restriction} for the notion of a restriction of a problem. 

\begin{lemma}\label{calim2section6}
Let $\Pi$ be a binary labeling problem that is not of type \tI, \tII, \tIII, \tIV, \tV{} or \tVI{}. Then there is a restriction $\Pi' \subseteq \Pi$ such that one of the following holds:
\begin{enumerate}[noitemsep]
    \item $\Pi'$ is $(2,1)$-resilient or $(1,2)$-resilient,
    \item $\Pi'$ is equivalent to a bipartite matching problem,
    \item $\Pi'$ is equivalent to a hypergraph matching problem,
    \item $\Pi'$ is equivalent to an edge grabbing problem.
\end{enumerate}
\end{lemma}
\begin{proof}
Let $\Pi=(d,\delta,W,B)$ be some problem that does not belong to type $\tI$--$\tVI$. If $\Pi$ is $(2,1)$-resilient or $(1,2)$-resilient, we are done. For the rest of the proof assume that $\Pi$ is neither $(2,1)$-resilient nor $(1,2)$-resilient. Therefore both of the following hold:
\begin{itemize}[noitemsep]
\item $W$ contains $\n^{d-1}$ or $B$ contains $\n^\delta$.
\item $W$ contains $\n^{d}$ or $B$ contains $\n^{\delta-1}$.
\end{itemize}
Now, we will go over all possible combinations and analyze them. First assume that $W$ contains $\n^d$. We have the following cases:
\begin{itemize}[noitemsep]
  \item $W=\n\n\nnn$: We have $\Pi \in \fIIa$.
  \item $W=\p\n\nnn$ and $B = \n\x\xxx$: We have $\Pi \in \fIa$.
  \item $W=\p\n\nnn$ and $B = \p\x\xxx$: We have $\Pi \in \fIVa$.
  \item $W=\n\nnn\p$ and $B = \x\xxx\n$: We have $\Pi \in \fIb$.
  \item $W=\n\nnn\p$ and $B = \x\xxx\p$: We have $\Pi \in \fIVb$.
\end{itemize}
But we assumed $\Pi \notin \tI, \tII, \tIV$, so none of these are possible.

Second, assume that $B$ contains $\n^\delta$. This is equivalent to $W$ containing $\n^d$; we would have $\Pi \in \tI$, $\Pi \in \tII$, or $\Pi \in \tIV$.

The final possibility is that $W$ does not contain $\n^d$, and $B$ does not contain $\n^\delta$.
Then $W$ has to contain $\n^{d-1}$, and $B$ has to contain $\n^{\delta-1}$, i.e.,
\begin{align*}
W&\in\{\p\nnn\p,\, \n\p\nnn,\, \nnn\p\n,\, \p\p\nnn,\, \nnn\p\p\},\\
B&\in\{\p\nnn\p,\, \n\p\nnn,\, \nnn\p\n,\, \p\p\nnn,\, \nnn\p\p\}.
\end{align*}
We go through all combinations of possible values of $W$ and $B$ in these sets:
\begin{itemize}[noitemsep]
  \item $W = \p\nnn\p$ and $B = \p\nnn\p$: We have $\Pi \in \fIVa$.
  \item $W = \p\nnn\p$ and $B = \n\p\nnn$:
  \begin{itemize}
    \item If $B = \n\p\n$, we have $\Pi \in \fVa$.
    \item Otherwise $B = \n\p\n\nnn$, and $\Pi$ is equivalent to a hypergraph matching problem.
  \end{itemize}      
  \item $W = \p\nnn\p$ and $B = \nnn\p\n$: Equivalent to the case $W = \p\nnn\p$ and $B = \n\p\nnn$ above.
  \item $W = \p\nnn\p$ and $B = \p\p\nnn$: We have $\Pi \in \fIVa$.
  \item $W = \p\nnn\p$ and $B = \nnn\p\p$: We have $\Pi \in \fIVb$.
  \bigskip
  \item $W = \n\p\nnn$ and $B = \p\nnn\p$: Equivalent to the case $W = \p\nnn\p$ and $B = \n\p\nnn$ above.
  \item $W = \n\p\nnn$ and $B = \n\p\nnn$:
  \begin{itemize}
    \item If $W = \n\p\n$, we have $\Pi \in \fVIa$.
    \item If $B = \n\p\n$, we have $\Pi \in \fVIb$.
    \item Otherwise $W = \n\p\n\nnn$ and $B = \n\p\n\nnn$, and $\Pi$ is equivalent to a bipartite matching problem.
  \end{itemize}      
  \item $W = \n\p\nnn$ and $B = \nnn\p\n$: We have $\Pi \in \fVIb$.
  \item $W = \n\p\nnn$ and $B = \p\p\nnn$:
  \begin{itemize}
    \item If $W = \n\p\n$ and $B = \p\p\nnn$, we have $\Pi \in \fVIa$.
    \item If $W = \n\p\n\nnn$ and $B = \p\p\n$, then $\Pi$ is equivalent to an edge grabbing problem.
    \item Otherwise $W = \n\p\n\nnn$ and $B = \p\p\n\nnn$. Define $\Pi' = (\n\p\n\nnn, \n\p\n\nnn)$. Now $\Pi' \subseteq \Pi$, and $\Pi'$ is equivalent to a bipartite matching problem.
  \end{itemize}      
  \item $W = \n\p\nnn$ and $B = \nnn\p\p$: We have $\Pi \in \fVIb$.
  \bigskip
  \item $W = \nnn\p\n$: Equivalent to one of the cases with $W = \n\p\nnn$ above.
  \bigskip
  \item $W = \p\p\nnn$ and $B = \p\nnn\p$: We have $\Pi \in \fIVa$.
  \item $W = \p\p\nnn$ and $B = \n\p\nnn$: Equivalent to $W = \n\p\nnn$ and $B = \p\p\nnn$ above.
  \item $W = \p\p\nnn$ and $B = \nnn\p\n$: We have $\Pi \in \fVIb$.
  \item $W = \p\p\nnn$ and $B = \p\p\nnn$: We have $\Pi \in \fIVa$.
  \item $W = \p\p\nnn$ and $B = \nnn\p\p$: We have $\Pi \in \fVIb$.
  \bigskip
  \item $W = \nnn\p\p$: Equivalent to one of the cases with $W = \p\p\nnn$ above. \qedhere
\end{itemize}
\end{proof}

\subsection{Rake \& compress procedures}

We now present two procedures, which will be are our main tool in order to show $O(\log n)$ upper bounds. These procedures are based on the rake and compress technique introduced by \citet{Miller1985}.

Let $G=(V,E)$ be a tree whose vertices are properly colored with two colors, called black and white. We start with two definitions; see Figure~\ref{fig:lowdegree-longpaths} for examples.
\begin{definition}[$\lowdegree$]\label{def:lowdegree}
Let $w,b \in \{1,2\}$. We define that $\lowdegree(G,w,b) \subseteq V$ is the set of all white nodes with degree at most $w$ and all black nodes with degree at most~$b$.
\end{definition}
\begin{definition}[$\longpaths$]\label{def:longpaths}
Let $p \in \{1,2,\dotsc\}$. Let $X \subseteq V$ consist of all nodes of degree exactly~$2$. We define that $\longpaths(G,p) \subseteq X$ consists of those nodes that belong to a connected component of size at least $p$ in the subgraph of $G$ induced by $X$.
\end{definition}

\begin{figure}
\centering
\includegraphics[page=4,scale=0.9]{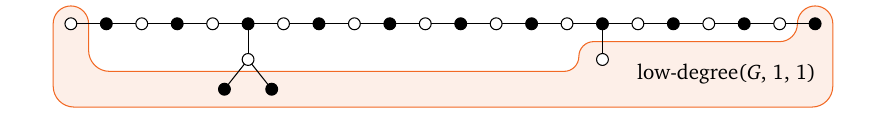}\\[7mm]
\includegraphics[page=5,scale=0.9]{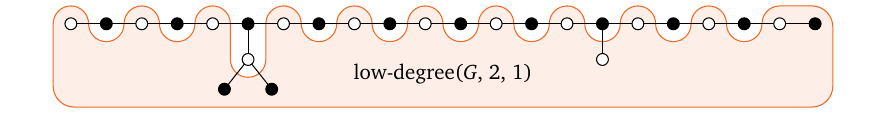}\\[7mm]
\includegraphics[page=6,scale=0.9]{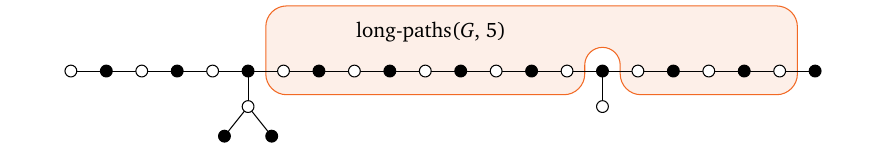}
\caption{Examples of $\lowdegree(G,w,b)$ from Definition~\ref{def:lowdegree} and $\longpaths(G,p)$ from Definition~\ref{def:longpaths}.}\label{fig:lowdegree-longpaths}
\end{figure}

Note that both of the definitions are local; $\lowdegree$ only depends on the degree of a node and to find out whether a given node $v$ is in $\longpaths(G,p)$, it is sufficient to explore its radius-$p$ neighborhood:
\begin{observation}\label{obs:lowedgree-longpaths}
For any constants $w,b,p$, there is a constant-round algorithm that finds which nodes belong to the set $\lowdegree(G,w,b)$ and which nodes belong to the set $\longpaths(G,p)$.
\end{observation}

Now we are ready to define two specific rake \& compress procedures, $\RC$ and $\RCP$. For a graph $G=(V,E)$ and a set of nodes $X$, we write $G-X$ for the subgraph of $G$ induced by $V \setminus X$, i.e., it is the subgraph formed by deleting nodes in $X$. See Figure~\ref{fig:RC-RCP} for illustrations.

\begin{definition}[$\RC$]\label{def:RC}
Let $w,b \in \{1,2\}$.
Procedure $\RC(w,b)$ partitions the set of nodes $V$ into non-empty sets $V_1, V_2, \dotsc, V_L$ for some $L$ as follows:
\begin{align*}
G_0 &= G, \\
V_{i+1} &= \lowdegree(G_i,w,b), & \text{$i = 0,1,\dotsc$ and $G_{i-1}$ is non-empty}, \\
G_{i+1} &= G_i - V_{i+1}, & \text{$i = 0,1,\dotsc$ and $G_{i-1}$ is non-empty}.
\end{align*}
\end{definition}

\begin{definition}[$\RCP$]\label{def:RCP}
Let $p \in \{1,2,\dotsc\}$.
Procedure $\RCP(p)$ partitions the set of nodes $V$ into non-empty sets $V_1, V_2, \dotsc, V_L$ for some $L$ as follows:
\begin{align*}
G_0 &= G, \\
V_{i+1} &= \lowdegree(G_i,1,1) \cup \longpaths(G_i,p), & \text{$i = 0,1,\dotsc$ and $G_{i-1}$ is non-empty}, \\
G_{i+1} &= G_i - V_{i+1}, & \text{$i = 0,1,\dotsc$ and $G_{i-1}$ is non-empty}.
\end{align*}
\end{definition}

\begin{figure}
\centering
\includegraphics[page=7,scale=0.9]{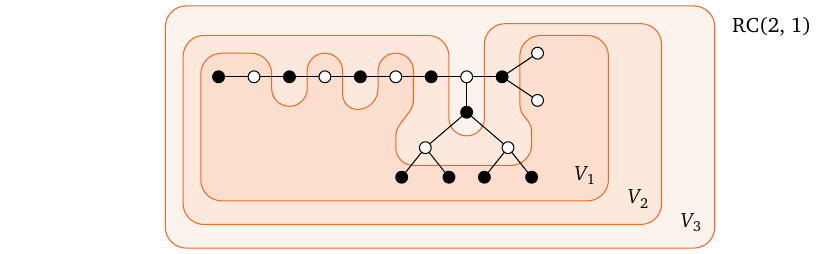}\\[7mm]
\includegraphics[page=8,scale=0.9]{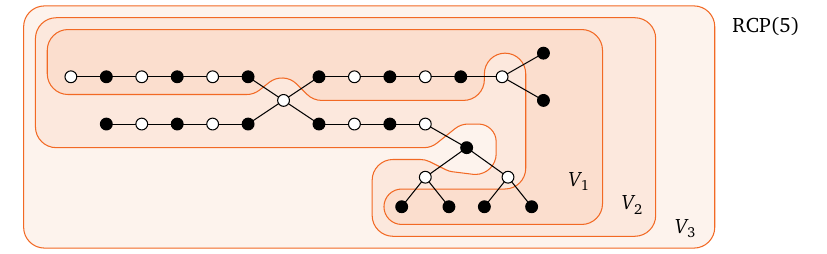}
\caption{Examples of processes $\RC(w,b)$ from Definition~\ref{def:RC} and $\RCP(p)$ from Definition~\ref{def:RCP}.}\label{fig:RC-RCP}
\end{figure}

We call the sets $V_i$ \emph{layers}. As we start with a tree $G_0$, each subgraph $G_i$ is a forest, and as long as it is non-empty there is at least one node of degree at most $1$. Therefore $V_{i+1}$ is also non-empty, and hence both $\RC$ and $\RCP$ eventually terminate after $L=O(n)$ layers. The key observation is that $L = O(\log n)$ layers suffices, as long as we do not use $\RC(1,1)$; this is what we will prove next.

\subsection{Rake \& compress runs in logarithmic time}

\begin{lemma}\label{lem:lowdegree11-longpaths}
Fix a $p = 1,2,\dotsc$, and let $G=(V,E)$ be a tree with $n$ nodes. At least one of the following holds:
\[
\bigl|\lowdegree(G,1,1)\bigr| \ge \frac{n}{6p} \quad\text{or}\quad
\bigl|\longpaths(G,p)\bigr| \ge \frac{n}{3}.
\]
\end{lemma}
\begin{proof}
If $n=1$, then $|\lowdegree(G,1,1)| = 1 = n$, and the claim holds. Otherwise $n > 1$, and the minimum degree of $G$ is $1$.

Let $n_1$ be the number of nodes of degree $1$, let $n_2$ be the number of nodes of degree $2$, and let $n_3$ be the number of nodes of degree at least $3$. We have $n = n_1+n_2+n_3$. The number of edges in tree $G$ is $m=n-1$, and by counting half-edges, we get
$n_1 + 2n_2 + 3n_3 \le 2m = 2n_1 + 2n_2 + 2n_3 - 2$.
Hence $n_3 < n_1$.

We have $|\lowdegree(G,1,1)| = n_1$, so if $n_1 \ge n/(6p)$, the claim holds. In what follows, assume that $n_1 < n/(6p)$. This implies that $n_1 + n_3 < 2n_1 < n/(3p)$.

Consider the subgraph $G_2$ induced by degree-$2$ nodes of $G$. If we contract each connected component of $G_2$ into an edge, we obtain a tree $G'$ in which we have $n' = n_1 + n_3$ nodes and $m' = n_1 + n_3 - 1$ edges, and each edge represents at most one connected component of $G_2$. Hence there are fewer than $n/(3p)$ components in $G_2$.

Components of size less than $p$ can contain fewer than $pn/(3p) = n/3$ nodes in total, but the total number of nodes of degree $2$ is $n_2 = n - (n_1 + n_3) > n - n/(3p) \ge 2n/3$. Therefore there have to be at least $n/3$ nodes in components of size at least $p$, and hence $|\longpaths(G,p)| \ge n/3$.
\end{proof}

\begin{lemma}\label{lem:lowdegree12}
Let $G=(V,E)$ be a tree with $n$ nodes. Then
\[
\bigl|\lowdegree(G,1,2)\bigr| \ge \frac{n}{12} \quad\text{and}\quad
\bigl|\lowdegree(G,2,1)\bigr| \ge \frac{n}{12}.
\]
\end{lemma}
\begin{proof}
Apply Lemma~\ref{lem:lowdegree11-longpaths} with $p=2$. If $|\lowdegree(G,1,1)| \ge n/(6p) = n/12$, we are done. Otherwise $|\longpaths(G,2)| \ge n/3$. Each connected component in the subgraph induced by $\longpaths(G,2)$ is a two-colored path of length at least $2$, and hence at least a fraction $1/3$ of the nodes in each such component is black. Therefore there are at least $n/9$ black nodes of degree $2$, and $|\lowdegree(G,1,2)| \ge n/9$. The same holds for white nodes and $|\lowdegree(G,2,1)|$.
\end{proof}

\begin{lemma}\label{generelizedRCP}
Procedures $\RC(1,2)$, $\RC(2,1)$, and $\RCP(p)$ terminate after $L = O(\log n)$ layers.  
\end{lemma}
 
\begin{proof}
We show that we eliminate at least a constant fraction of nodes per step. We can assume that $G_i$ is connected; otherwise we can analyze each connected component separately. For $\RC(w,b)$, we can apply Lemma~\ref{lem:lowdegree12} to see that we eliminate at least a fraction $1/12$ of the nodes in each step. For $\RCP(p)$, we can apply Lemma~\ref{lem:lowdegree11-longpaths} to see that we eliminate at least a fraction $1/(6p)$ of the nodes in each step.
\end{proof}

\begin{corollary}\label{cor:RCPimpl}
Procedures $\RC(1,2)$, $\RC(2,1)$, and $\RCP(p)$ can be implemented in $O(\log n)$ rounds in the LOCAL model of computation.
\end{corollary}
\begin{proof}
Follows from Observation~\ref{obs:lowedgree-longpaths} and Lemma~\ref{generelizedRCP}: it is sufficient to repeat a constant-time step for $O(\log n)$ times.
\end{proof}

\subsection{Properties of layer decompositions}

Given a decomposition of nodes in layers $V_1, V_2, \dotsc, V_L$ and a node $v \in V$, we define
\begin{align*}
N(v,i) &= \bigl\{ u \in V_i : \{u,v\} \in E \bigr\}, &
d(v,i) &= |N(v,i)|, \\
N(v,i\ldots j) &= N(v,i) \cup N(v,i+1) \cup \dotsb \cup N(v,j), &
d(v,i\ldots j) &= |N(v,i\ldots j)|.
\end{align*}
That is, $N(v,i\ldots j)$ is the set of neighbors of $v$ that are on layers $i,i+1,\dotsc,j$.

\begin{claim}\label{claim:RCdegree}
In a decomposition $V_1,V_2,\dots,V_L$ produced by $\RC(w,b)$, we will have $d(v,i\ldots L) \le w$ for all white nodes $v \in V_i$ and $d(v,i\ldots L) \le b$ for all black nodes $v \in V_i$.
\end{claim}
\begin{proof}
Note that $d(v,i\ldots L)$ is equal to the number of neighbors of $v$ in graph $G_{i-1}$, and by assumption we had $v \in V_i = \lowdegree(G_{i-1},w,b)$, so if $v$ is white, there are at most $w$ neighbors in $G_{i-1}$, and if $v$ is black, there are at most $b$ neighbors in $G_{i-1}$.
\end{proof}

\begin{claim}\label{claim:RCPdegree}
In a decomposition $V_1,V_2,\dots,V_L$ produced by $\RCP(p)$, $p > 1$, we will have
\[d(v,i\ldots L) \le 2 \quad\text{and}\quad d(v,i+1\ldots L) \le 1\]
for all nodes $v \in V_i$.
\end{claim}
\begin{proof}
We have $v \in V_i$, and hence one of the following holds:
\begin{itemize}
  \item $v \in \lowdegree(G_{i-1},1,1)$. Then $v$ has at most $1$ neighbor in $G_{i-1}$, and hence
  \[
    d(v,i+1\ldots L) \le d(v,i\ldots L) \le 1.
  \]
  \item $v \in \longpaths(G_{i-1},p)$. Then $v$ has at most $2$ neighbors in $G_{i-1}$. Furthermore, $v$ is a part of a connected component of size at least $p$, and the entire component belongs to $V_i$; hence there is a neighbor $u$ of $v$ in $V_i$, and in particular $u \in N(v,i\ldots L)$ but $u \notin N(v,i+1\ldots L)$. We have
  \begin{align*}
    d(v,i\ldots L) &= 2, \\
    d(v,i+1\ldots L) &\le d(v,i\ldots L) - 1 = 1. \qedhere
  \end{align*}
\end{itemize}
\end{proof}

\subsection{Algorithm for resilient problems}\label{ssec:-resilupperbound}

Now we are ready to use $\RC$ and $\RCP$ to develop $O(\log n)$-time algorithms for all four problem classes that we defined in Section~\ref{ssec:log-upper-def}. We will start with resilient problems in this section; we will discuss bipartite matching, hypergraph matching, and edge grabbing later.

In what follows, we will often deal with \emph{partial} labelings of graphs:
\begin{definition}[Partial labeling]
A \emph{partial labeling} of a tree $G = (V,E)$ is a labeling $g\colon E \to \{0,1,\bot\}$, where $g(e) = \bot$ indicates that edge $e$ is unlabeled. Given two partial labelings $g$ and $g'$, we say that $g'$ is a \emph{completion} of $g$ if $g'(e) = \bot$ implies $g(e) = \bot$ and $g(e) \ne \bot$ implies $g'(e) = g(e)$; that is, $g'$ is constructed from $g$ by changing some empty labels to non-empty labels. We say that the labeling around $v$ is complete in $g$ if all edges incident to $v$ have non-empty labels.

Given a problem $\Pi = (d,\delta,W,B)$ and a partial labeling $g$, we say that the labeling around $v$ is \emph{valid} in $g$ if it is complete around $v$ and $v$ satisfies the constraints of $\Pi$, i.e., the number of incident edges labeled with $1$ is in $W$ if $v$ is a white node and it is in $B$ if $v$ is a black node. Finally, we say that a labeling $g$ \emph{can be completed} around $v$ if there exists a completion $g'$ of $g$ such that the labeling around $v$ is valid in $g'$. We extend these definitions to sets of nodes in a natural manner.
\end{definition}

Note that if $v$ is not a relevant node (recall that only white nodes of degree $d$ and black nodes of degree $\delta$ are relevant), then any partial labeling can be trivially completed around $v$.

Now we are ready to connect the definition of resilience with the definition of completability:

\begin{lemma}\label{lem:resilience-complete}
Let $\Pi=(d,\delta,W,B)$ be a $(t,s)$-resilient problem, let $G = (V,E)$ be a tree, and let $g$ be a partial labeling of $G$. Then $g$ is completable around any $v \in V$ that satisfies one of the following:
\begin{itemize}[noitemsep]
\item $v$ is a white node incident to at most $t$ edges with non-empty labels,
\item $v$ is a black node incident to at most $s$ edges with non-empty labels.
\end{itemize}
\end{lemma}
\begin{proof}
We prove the claim for white nodes; the case of black nodes is analogous. Let $v$ be a white node incident to at most $t$ edges with non-empty labels. If the degree of $v$ is not $d$, the claim trivially holds, so assume that $v$ is a node of degree $d$.

Let $i \le t$ be the number of incident edges that are labeled with $1$. By the definition of resilience, there has to be an index $j \in W$ such that
\[
j \in \{ i, i+1, \dots, i+d-t \}.
\]
If this was not the case, we would have $d+1-t$ consecutive indexes $j$ with $j \notin W$, i.e., a substring of the form $\n^{d+1-t}$ in the bit string representation of set $W$. Note that $j-i \le d-t$ and therefore there are at least $j-i$ unlabeled edges around $v$.

This gives now a simple way to find a valid labeling for the unlabeled edges around $v$: take the first $j-i$ unlabeled edges incident to $v$, label them with $1$, and label all other edges incident to $v$ with $0$.
\end{proof}

\begin{remark}
The converse of Lemma~\ref{lem:resilience-complete} also holds---it turns out that resilience is \emph{equivalent} to completability. However, one direction is sufficient for us.
\end{remark}

\begin{figure}
\centering
\includegraphics[page=9,scale=0.9]{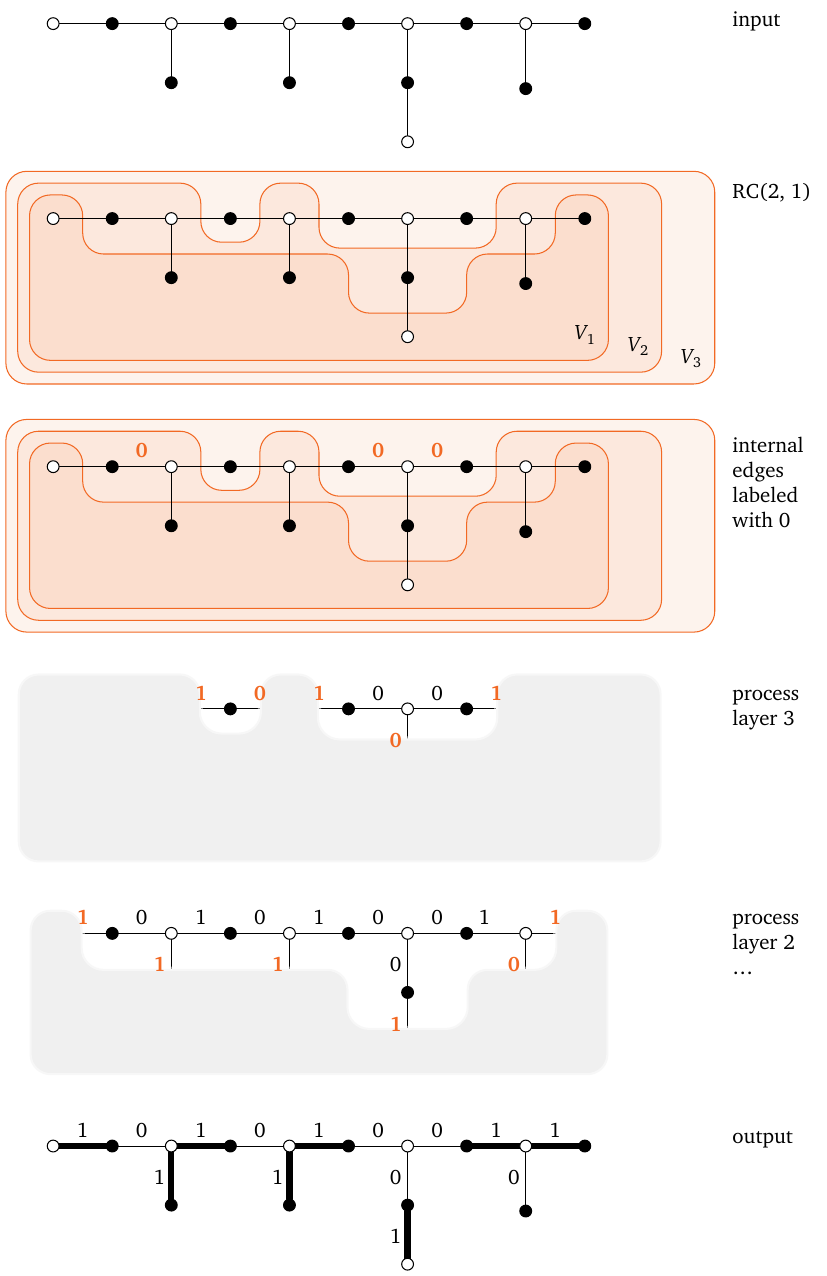}
\caption{Illustration of the proof of Theorem~\ref{theorem21resil}: solving $(2,1)$-resilient problems in $O(\log n)$ rounds, with the help of $\RC(2,1)$ procedure. Here we use the \emph{even orientation} problem (Example~\ref{ex:even-orientation}) as an example---the constraints are $W = \p\n\p\n$ and $B = \n\p\n$, and the constraint need to be satisfied for white nodes of degree $d=3$ and black nodes of degree $\delta=2$.}\label{fig:even-orientation}
\end{figure}

\begin{theorem}\label{theorem21resil}
If $\Pi=(d,\delta,W,B)$ is $(2,1)$-resilient or $(1,2)$-resilient, then $\Pi$ can be solved in $O(\log n)$ rounds.
\end{theorem}
\begin{proof}
Without loss of generality, let us consider a $(2,1)$-resilient problem. We use the rake and compress procedure $\RC(2,1)$ to construct a layer decomposition $V_1,V_2,\dotsc,V_L$; we have $L=O(\log n)$ due to Lemma \ref{generelizedRCP} and this takes $O(\log n)$ rounds due to Corollary \ref{cor:RCPimpl}.

We then use the layer decomposition to construct a feasible solution $X$ to $\Pi$; see Figure~\ref{fig:even-orientation} for an example. First, edges between nodes of the same layer get labeled with $0$. Then, nodes proceed sequentially by layer number, from layer $L$ to layer $1$.
For each layer $i$, assume that all nodes at higher layers have already labeled their incident edges in a valid manner. We prove that we can complete the labeling of the edges incident to nodes of layer $i$ in a valid manner.

Consider a white node $v$ at layer $i$. By Claim~\ref{claim:RCdegree}, $d(v,i\ldots L) \le 2$, and hence at most two edges incident to $v$ have already fixed their labels. Since the problem is $(2,1)$-resilient, by Lemma~\ref{lem:resilience-complete} white nodes can complete any such partial labeling in their neighborhood.

Similarly, black nodes have at most one edges incident to them that are already labeled, and as the problem is $(2,1)$-resilient, black nodes can also complete any such partial labeling in their neighborhood. Note that the choices of two adjacent nodes that are on the same layer do not interfere with each other, as all layer-internal edges are fixed to label $0$.  Further, in any layer, particularly in layer $1$, that contains all the leaves, there might white nodes of degree $<d$, and black nodes of degree $<\delta$. These nodes are always satisfied and can always complete their labeling in an arbitrary manner.

We need $O(1)$ rounds per layer and hence $O(\log n)$ rounds in total to construct a feasible solution to $\Pi$.
\end{proof}

\subsection{Algorithm for bipartite matching problems}\label{ssec:bipartite-matching}

\begin{theorem}[Bipartite matching]\label{PMatchingTheorem}
Problems of the form $\Pi=(\n\p\n\nnn, \n\p\n\nnn)$ can be solved in $O(\log n)$ rounds.
\end{theorem}
\begin{proof}
We apply the procedure $\RCP(3)$ to the input tree $G$. Denote by $V_1,\dotsc,V_L$ the resulting decomposition. We process the layers one by one, from layer $L$ to $1$.

Consider the nodes on layer $i$, and assume that all nodes of layers higher than $i$ have labeled their incident edges in a valid manner. We show that nodes of layer $i$ can label their incident edges as well. Recall that only white nodes of degree $d$ and black nodes of degree $\delta$ are \emph{relevant} and all other nodes are unconstrained. Note that in bipartite matching problems, all relevant nodes have degree at least $3$.

By Claim~\ref{claim:RCPdegree} each node $v$ at layer $i$ satisfies $d(v,i\ldots L) \le 2$ and $d(v,i+1\ldots L) \le 1$. Hence relevant nodes have at most one incident edge fixed by higher-level nodes, and they have at least one edge pointing towards lower-level nodes. In particular, relevant nodes are currently incident to at most one edge with label $1$, and each relevant node $v$ at level $i$ can pick one incident edge $e = \{v,u\}$ such that $u$ is on a lower level. Node $v$ then labels all other unlabeled incident edges with $0$, and finally it assigns $0$ or $1$ to edge $e$ in order to satisfy its constraints. Nodes that are not relevant can simply assign $0$ to all incident unlabeled edges.

Note that the choices of the nodes on the same level do not conflict with each other, as all edges between nodes of level $i$ are labeled with $0$. Thus we can safely label each layer in $O(1)$ rounds.
\end{proof}

\begin{figure}
\centering
\includegraphics[page=10,scale=0.9]{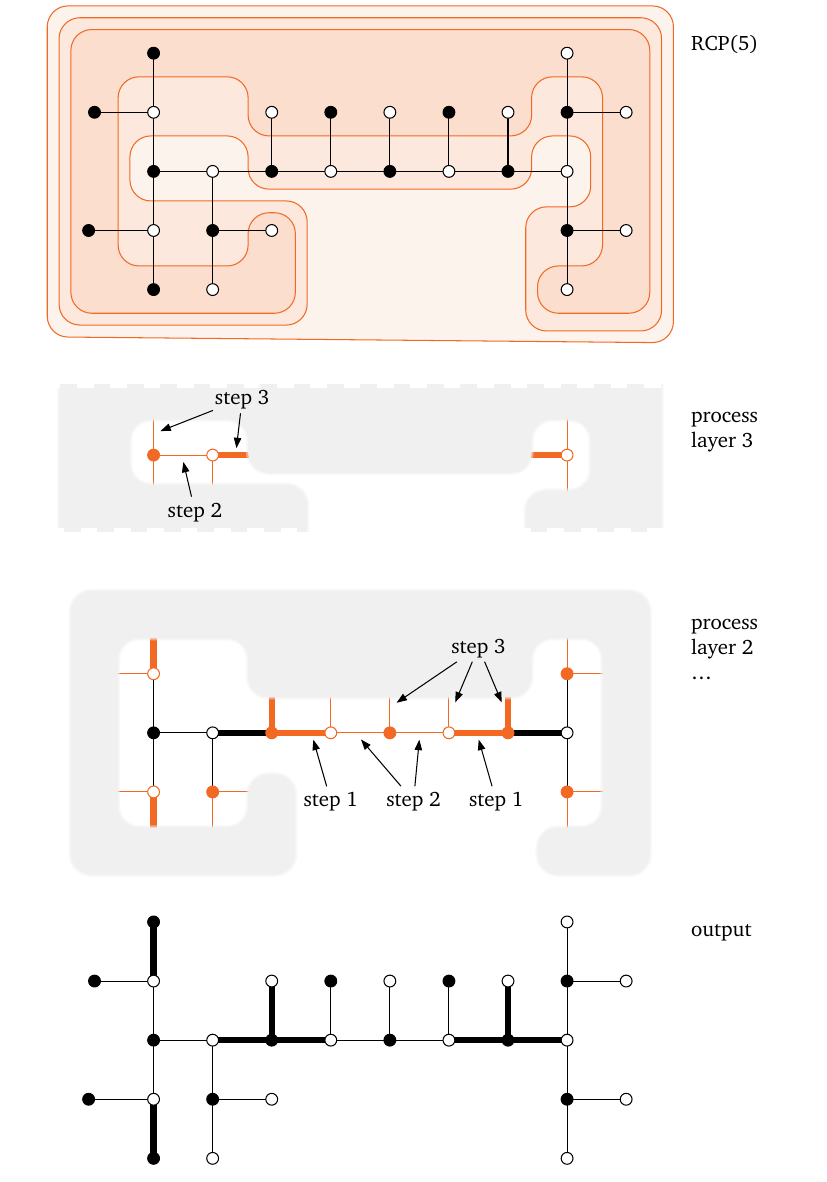}
\caption{Illustration of the proof of Theorem~\ref{whiteoneBlackAllOrNothing}: solving hypergraph matching in $O(\log n)$ rounds, with the help of $\RCP(5)$ procedure. Here we use the case of $d=\delta=3$ as an example, i.e., the problem to solve is $W=\n\p\n\n$, $B=\p\n\n\p$.}\label{fig:hypergraph-matching}
\end{figure}

\subsection{Algorithm for hypergraph matching problems}\label{ssec:hypergraph-matching}

\begin{theorem}[Hypergraph matching]\label{whiteoneBlackAllOrNothing}
Problems of the form $\Pi = (\n\p\n\nnn, \p\nnn\p)$ can be solved in $O(\log n)$ rounds.
\end{theorem}
\begin{proof}
We apply the procedure $\RCP(5)$ to the input tree $G$ to construct a decomposition $V_1,\dotsc,V_L$; see Figure~\ref{fig:hypergraph-matching} for an example. This choice of the parameter implies that the connected components of $G[V_i]$ are isolated nodes and paths of length at least $5$.  Recall that Claim~\ref{claim:RCPdegree} shows that each node $v \in V_i$ has at most two neighbors in layers $i,i+1,\dotsc,L$ and at most one neighbor in layers $i+1,\dotsc,L$.

We iterate through the layers in a backwards order and in iteration $i=L,\ldots,1$ the vertices in $V_i$ use a three step procedure to label all their incident edges that have not been labeled in any of the previous iterations:
\begin{itemize}
\item\textbf{Step 1}: Assume that $v$ is a black node, it is adjacent to a higher-layer node $u \in N(v,i+1\ldots L)$, and the edge $\{v,u\}$ is already labeled with $1$. Now if $v$ is adjacent to another node $s \in V_i$ on the same layer, we label the edge $\{v,s\}$ with $1$.
\item\textbf{Step 2:} Label all edges between vertices in $V_i$ that are unlabeled after step 1 with $0$.
\item\textbf{Step 3:} Each vertex completes the labeling of its incident edges going to lower layers in an arbitrary valid way.
\end{itemize}
Vertices only need to know their $O(1)$ neighborhood to execute the three steps. We show that each vertex can complete its labeling in step 3 in a valid way:
\begin{itemize}
  \item \textit{Black nodes:} A black node $v$ either has all of its already labeled edges labeled with a $0$ or it has two edges incident edges labeled with a $1$. Thus any black node can complete the labeling of its incident edges towards lower levels with all $0$s in the former and with all $1$s in the latter case.
  \item \textit{White nodes with higher-level neighbors:} Consider a white node $v\in V_i$ and assume that $v$ has a higher-layer neighbor $u$. Then $v$ cannot have a black neighbor $v_1$ with a higher-layer neighbor $u_1$, as otherwise $v$ and $v_1$ would only form a connected component that is a path of length $2$, a contradiction with the fact that $G[V_i]$ consist of isolated nodes and paths of length at least $5$. Thus $v$ has at most one incident edge labeled with a $1$ after step 2, i.e., the edge $\{v,u\}$, and $v$ can complete the labeling. 
  \item \textit{White nodes without higher-level neighbors:} Consider a white node $v\in V_i$ and assume that $v$ has no neighbors on higher layers. Then $v$ can have at most one black neighbor that participated in step $1$ because if it had two such neighbors $v_1$ and $v_2$, then $v_1$, $v$ and $v_2$ would again induce a connected component that is a path of length $3$ in $G[V_i]$, a contradiction. Thus $v$ has at most one black neighbor $v_1$ that potentially labeled the edge $\{v,v_1\}$ with a $1$, all other already labeled incident edges of $v$ must be labeled with $0$, that is, $v$ can complete the labeling. \qedhere
\end{itemize}
\end{proof}

\subsection{Algorithm for edge grabbing problems}\label{ssec:edge-grabbing}

We next show that any problem in the edge grabbing family can be reduced to sinkless orientation which implies that it can be solved in $O(\log n)$ rounds.

\begin{theorem}[Edge grabbing]\label{edgeGrabbingThrm}
Problems of the form $\Pi = (\n\p\n\nnn,\p\p\n)$ can be solved in $O(\log n)$ rounds.
\end{theorem}
\begin{proof}
We will show this upper bound by reducing the problem to sinkless orientation, that is, we reduce $\Pi$ to a problem of the form $\Pi'=(\n\p\p\ppp, \n\p\n)$, which can be solved in $O(\log n)$ rounds by Theorem~\ref{theorem21resil}.

Assume we are given a solution to $\Pi'$. Each white node $v$ of degree $d$ is incident to at least one edge that is labeled with $1$. Node $v$ chooses one such edge (say, the one with the smallest port number) and marks it \emph{red}. All other edges are \emph{black}.

Let $X$ be the set of red edges. These edges form a feasible solution to $\Pi$: each white node of degree $d$ is incident to exactly one red edge, and each black node of degree $2$ is incident to at most one red edge.
\end{proof}

\subsection{Proof of Theorem \ref{thm:log-upper}}

Lemma \ref{calim2section6} shows that any problem that is not of type $\tI$, $\tII$, $\tIII$, $\tIV$, $\tV$ or $\tVI$ falls in one of the following four classes:
\begin{enumerate}[noitemsep]
    \item $(2,1)$-resilient problems and $(1,2)$-resilient problems.
    \item Bipartite matching and equivalent problems.
    \item Hypergraph matching and equivalent problems.
    \item Edge grabbing and equivalent problems.
\end{enumerate}
Theorems \ref{theorem21resil}, \ref{PMatchingTheorem}, \ref{whiteoneBlackAllOrNothing}, and \ref{edgeGrabbingThrm} show that problems in each of the four classes can be solved in $O(\log n)$ rounds. This completes the proof.

\section{Logarithmic lower bounds}\label{sec:log-lower}

In this section we will show that all binary labeling problems that are not of the type discussed in Section~\ref{sec:unsol}, Section~\ref{sec:constant}, or Section~\ref{sec:global} require $\Omega(\log n)$ rounds with deterministic distributed algorithms. This will give the lower bounds for the final section of Table~\ref{tab:deterministic}, and hence it completes the classification of all binary labeling problems. Formally we prove the following theorem.
\begin{theorem}\label{thm:log-lower}
Any problem that is not of type \tI, \tII, \tIII, \tIV, \tV{} or \tVI{} has complexity $\Omega(\log n)$.
\end{theorem}

To prove Theorem~\ref{thm:log-lower}, we first introduce two families of problems: \emph{forbidden degree problems} and \emph{sinkless orientation problems}. In Lemma~\ref{lem:classificationLB} we will prove that any problem not covered in Sections \ref{sec:unsol}--\ref{sec:global} is at least as hard as a forbidden degree problem or a sinkless orientation problem. Then we will prove a lower bound of $\Omega(\log n)$ for both problem families.

\subsection{Definitions}

\begin{definition}[Forbidden degree]\label{def:forbidden-degree}
Problems of the form $\Pi = (\ppp\n\ppp,\n\ppp\n)$ are called \emph{forbidden degree problems}.
\end{definition}
\begin{remark}
Note that here we can have any degrees $d, \delta \geq 2$.
In forbidden degree problems, there is one value $i \notin W$, $0 < i < d$.
This value is known as the forbidden degree; the intuition here is that white nodes can have any $X$-degree except $i$.
\end{remark}

\begin{definition}[Bipartite sinkless orientation]\label{def:bipartite-so}
Problems of the form $\Pi = (\p\ppp\n,\n\p\ppp)$ are called \emph{bipartite sinkless orientation problems}.
\end{definition}

\begin{remark}
Note that here we can have any degrees $d, \delta \geq 2$.
For bipartite sinkless orientation, it is known~\cite{Brandt2016,chang16exponential,chang18complexity} that the same lower bounds of $\Omega(\log \log n)$ and $\Omega(\log n)$ as above hold for the case $d, \delta \geq 3$, and it is not hard to extend these bounds to the case where $d = 2$ or $\delta = 2$.
However, to avoid case distinctions and make sure that the bounds also hold in our setting (where, e.g., we have a tree instead of a general graph), we will prove these lower bounds for bipartite sinkless orientation in an analogous fashion to the lower bound proofs for the forbidden degree problems.
\end{remark}

\subsection{Coverage}

We will first show that all problem that are not unsolvable (types $\tI$ and $\tII$) or trivial (types $\tIII$ and $\tIV$) are at least as hard as forbidden degree problems or bipartite sinkless orientation problems:

\begin{lemma}\label{lem:classificationLB}
Let $\Pi=(d,\delta,W,B)$ be a problem that is not of type $\tI$, $\tII$, $\tIII$ or $\tIV$. Then there is a relaxation $\Pi' \supseteq \Pi$ such that $\Pi'$ is equivalent to a forbidden degree problem or a bipartite sinkless orientation problem.
\end{lemma}
\begin{proof}
As before, we denote the bits in $W$ and $B$ by $w_0, w_1, \dots, w_d$ and $b_0, b_1, \dots, b_\delta$, respectively. Consider all possible $4$-bit strings $w_0\,b_0\,w_d\,b_{\delta}$:
\begin{itemize}
\item $\p\p\x\x$: We have $\Pi = (\p\x\xxx, \p\x\xxx) \in \fIVa$.
\item $\x\x\p\p$: We have $\Pi = (\x\xxx\p, \x\xxx\p)$, which is equivalent to $\Pi = (\p\x\xxx, \p\x\xxx)$ above.
\item $\n\x\x\n$: We have $\Pi = (\n\x\xxx, \x\xxx\n) \subseteq \Pi' = (\n\p\ppp, \p\ppp\n)$.\\
\phantom{$\n\x\x\n$: }Now $\Pi'$ is a bipartite sinkless orientation problem.
\item $\x\n\n\x$: We have $\Pi = (\x\xxx\n, \n\x\xxx)$, which is equivalent to $\Pi = (\n\x\xxx, \x\xxx\n)$ above.
\item $\p\n\p\n$: We have $\Pi = (\p\xxx\p, \n\xxx\n)$.\\
\phantom{$\p\n\p\n$: }If $\Pi = (\p\p\ppp, \n\xxx\n)$, then $\Pi \in \fIIIb$.\\
\phantom{$\p\n\p\n$: }Otherwise $\Pi \subseteq \Pi' = (\ppp\n\ppp, \n\ppp\n)$, and $\Pi'$ is a forbidden degree problem.
\item $\n\p\n\p$: We have $\Pi = (\n\xxx\n, \p\xxx\p)$, which is equivalent to $\Pi = (\p\xxx\p, \n\xxx\n)$ above.
\qedhere
\end{itemize}
\end{proof}

\subsection{Automatic round elimination} \label{ssec:automatic-simulation}

We will show in this section that all forbidden degree problems require $\Omega(\log \log n)$ randomized time in the LOCAL model, and the same holds also for bipartite sinkless orientation. To prove the claim, we will use \emph{round elimination}. We begin by defining the necessary concepts.

The various round elimination proofs~\cite{Naor1991,Brandt2016,chang18complexity,Balliu2019weak2col,Balliu2019,Brandt2019automatic,Laurinharju2014} have all used slightly different notation. In this section we adopt the framework of Brandt~\cite{Brandt2019automatic} and much of the notation of Balliu et al.\ \cite{Balliu2019}.

\paragraph{Model of computing.} We use the randomized model of computing from Section~\ref{sec:model}, with two important distinctions.
\begin{enumerate}
  \item We always assume that the algorithm is running in the middle of a 2-colored tree with white degree $d$ and black degree $\delta$. In particular, we assume that the nodes never see any irregularities, such as leaves. This implies that the topology of the input network is fixed. We show that the problem is hard even with this promise.
  \item We assume that the nodes have random port numbers assigned independently from their random bits. This assumption simplifies the analysis of Lemma~\ref{lem:rand-generalized-simulation}. This model can simulate the randomized LOCAL model, since there an algorithm must work with high probability for \emph{any} fixed port numbering.
\end{enumerate}

\paragraph{Bipartite labeling problems.} In this section we consider a generalized version of binary labeling problems with $|\Sigma| > 2$. A \emph{bipartite labeling problem} is given by a tuple $(\Sigma,d,\delta,W,B)$, where $\Sigma$ is the alphabet, $d$ is the white degree, $\delta$ is the black degree, $W$ is the set of \emph{white configurations}, and $B$ is the set of \emph{black configurations}.
Here, each white configuration is a multiset containing exactly $d$ (not necessarily distinct) labels from $\Sigma$, and each black configuration is a multiset containing exactly $\delta$ (not necessarily distinct) labels from $\Sigma$.
The problem is solved correctly if each edge is labeled with some $\sigma \in \Sigma$, and for each white, resp.\ black, node, the multiset $\{ x_1, x_2, \dots, x_k \}$, where $k=d$, resp.\ $k=\delta$, of incident edge labels is a configuration contained in $W$, resp.\ $B$.

We represent (sets of) configurations as regular expressions~\cite{Balliu2019}. We use the shorthands $[x_1 x_2 \cdots x_k ] = x_1 | x_2 | \cdots | x_k$ and $x^k$ for $x$ repeated $k$ times, e.g., the regular expression $XY^2[XZ]$ stands for the set containing the two multisets $\{ X, X, Y, Y\}$ and $\{ X, Y, Y, Z \}$.

\paragraph{Output problems.} For each bipartite labeling problem $\Pi = (\Sigma, d, \delta, W, B)$ there exists a problem $\Pi' = (2^{\Sigma}, d, \delta, W', B')$ such that if $\Pi$ can be solved in $T$ rounds by a white algorithm, then $\Pi'$ can be solved in $T-1$ rounds by a black algorithm (with slightly worse success probability). We will provide a definition for such a problem below and call the problem $\Pi'$ specified by the definition the \emph{black output problem of $\Pi$}.
Similarly, if $\Pi$ can be solved in $T$ rounds by a black algorithm, we define a \emph{white output problem of $\Pi$} that can be solved by a $(T-1)$-round white algorithm.
Note that the definitions of the black and white output problems of $\Pi$ depend only on $\Pi$ and not on $T$.

In the following we define the black output problem $\Pi'$ of $\Pi$. The white output problem is obtained by applying the same operations with the roles of white and black reversed.

Given $\Pi$, we form the black constraint $B'$ of $\Pi'$ as follows. A multiset $X_1 X_2 \dots X_{\delta}$, where $X_i \subseteq 2^{\Sigma}$, is in $B'$ if and only if for all $(x_1, x_2, \dots, x_{\delta}) \in X_1 \times X_2 \times \cdots \times X_{\delta}$ we have that $x_1 x_2 \dots x_{\delta}$ is in $B$.
The white constraint is formed as follows. A multiset $Y_1 Y_2 \dots Y_d$, where $Y_i \subseteq 2^{\Sigma}$, is in $W'$ if and only if there exist $y_1 \in Y_1, y_2 \in Y_2, \dots, y_d \in Y_d$ such that $y_1 y_2 \dots y_d$ is in $W$.

Given a solution to the black (resp.\ white) output problem of $\Pi$, the white (resp.\ black) nodes can solve the initial problem $\Pi$ in one additional round by construction~\cite{Brandt2019automatic}.
We show essentially the converse in Lemma~\ref{lem:rand-generalized-simulation}.

\paragraph{Simplification via maximality.} We simplify any black output problem $\Pi' = (2^{|\Sigma|},d,\delta,W',B')$ using an observation of Brandt~\cite[Section 4.2]{Brandt2019automatic}. We consider only \emph{maximal} black configurations: if there are black configurations $Y_1 Y_2 \cdots Y_\delta$ and $Z_1 Z_2 \cdots Z_\delta$ such that $Y_i \subseteq Z_i$ for all $i$, then we can safely remove the constraint $Y_1 Y_2 \cdots Y_\delta$ from $B'$ without changing the (black) complexity of the problem. In particular, any black algorithm $A$ solving $\Pi'$ can always safely output $Z_1 Z_2 \cdots Z_\delta$ instead of $Y_1 Y_2 \cdots Y_\delta$ since this can only increase the chances that the white constraint is satisfied at a white node, according to our definition of a black output problem. Moreover, we can remove all white configurations that contain labels (from $2^\Sigma$) that are not contained in any of the black configurations (after applying the above maximality argument), as those labels cannot occur in a correct output. We will use these simplifications (and the analogous versions for white output problems) throughout the proofs of Lemmas~\ref{lem:speedup-part-1}~and~\ref{lem:speedup-part-2}.

\paragraph{Local failure probability.} Way say that an algorithm $A$ has local failure probability at most $p$ if for all nodes, the probability that the labeling incident to that node does not satisfy the constraint of that node is at most $p$. Since randomness is the only input, all white nodes have degree $d$, all black nodes have degree $\delta$, and the graph is a tree, all white nodes and all black nodes in the middle of the tree, respectively, have the same local failure probability.

\subsection{Local success probability decay}

In this section we show how to bound the failure probability of the round elimination. This proof is essentially present in previous round elimination arguments, but has always been presented in a problem-specific way. Lemma~\ref{lem:rand-generalized-simulation} presents it here in full generality. In particular, our argument does not depend on the particular structure of the LCL, only the number of different output labels and the degrees of the nodes.

\begin{lemma} \label{lem:rand-generalized-simulation}
  Let $A$ be a white $T$-round anonymous randomized algorithm for a bipartite labeling problem $\Pi$ with local failure probability at most $p$ in the randomized LOCAL model. Then there exists a black $(T-1)$-round algorithm $A'$ for $\Pi'$, the black output problem of $\Pi$, with local failure probability at most $Kp^{1/(\delta+1)}$, for some constant $K = K(\delta, d, |\Sigma|)$. Respectively, for a black $T$-round algorithm for $\Pi$ with local failure probability $p$ there exists a white $(T-1)$-round algorithm for the white output problem of $\Pi$ with local failure probability at most $Kp^{1/(d+1)}$ for some constant $K = K(d, \delta, |\Sigma|)$.
\end{lemma}

\begin{proof}
  We prove the case of initial white algorithm. The case of a black algorithm is symmetric. We connect the failure probability of algorithm $A'$ to the failure probability of the original algorithm $A$. We denote the $T$-hop ball around node $v$, considered as a graph, by $B_T(v)$. We denote the labeled graph obtained by labeling $B_T(v)$ with the inputs of each node by $I_T(v)$.

  Given a $T$-round white algorithm $A$, black nodes can execute a $(T-1)$-round \emph{simulation} algorithm $A'$ as follows. Each black node $v$ looks at the input, including the random inputs and the random port numbers, inside its $(T-1)$-neighborhood. It simulates $A$ at each of its neighbors, treating them as the active nodes, over all possible assignments of random inputs outside $B_{T-1}(v)$. We observe that given $I_{T-1}(v)$, the outputs of the neighbors of $v$, as random variables, are independent from each other: for all $u, w \in N(v)$ we have that $(B_T(u) \cap B_T(w)) \setminus B_{T-1}(v)$ is empty. This follows from the fact that we assumed the input graph to be a tree and implies independence since we assumed that the nodes do not have any correlated inputs, in particular that the nodes do not have unique identifiers. Each black node $v$ can independently compute the conditional probabilities $\Pr[A(u,v) = \sigma \mid I_{T-1}(v)]$ for all labels $\sigma$ and neighbors $u$, that is, the probability of $u$ outputting $\sigma$ on edge $\{u,v\}$ as an active node executing algorithm $A$ after fixing $I_{T-1}(v)$.

  We say that an output label $\sigma$ is \emph{frequent} in $A$ at node $u$ for $v$ if
  \[
    \Pr\bigl[A(u,v) = \sigma \bigm| I_{T-1}(v)\bigr] \geq f.
  \]
  We will set the value of $f$ later. The output of the simulation algorithm $A'$ for edge $\{v,u\}$ at $v$, denoted by $A'(v,u)$, is defined as the set of frequent output labels, that is,
  \[
    A'(v,u) = S_v(u) = \{ \sigma \in \Sigma \mid \text{$\sigma$ is frequent in $A$ at $u$ for $v$} \}.
  \]
  Note that $v$ can determine $A'(v,u)$ for each neighbor $u$ based on $I_{T-1}(v)$.

  We want to consider only the random assignments $I_{T-1}(v)$ that make $A$ likely to succeed, and discard the tail of assignments that cause $A$ to fail with a large probability. We say that the $(T-1)$-neighborhood $I_{T-1}(v)$ of a black node $v$ is \emph{good} if $A$ is likely to succeed around it. Let $E(v)$ denote the event that for a black node $v$ the outputs $A(u,v)$ for all $u \in N(v)$ form a black configuration in $\Pi$. For a parameter $r$ to be fixed later, we say that $I_{T-1}(v)$ is good if and only if $\Pr[E(v) \mid I_{T-1}(v)] \geq 1-r$. Since the outputs of the neighbors are independent, this probability can again be determined by $v$ based on $I_{T-1}(v)$.

  Most random assignments $I_{T-1}(v)$ must be good, as otherwise the original algorithm $A$ would not work with probability $1-p$: we have that $r \cdot \Pr[I_{T-1}(v) \text{ is not good}] \leq p$. We get that
  \begin{align*}
    \Pr[\text{$I_{T-1}(v)$ is good}] &= 1 - \Pr[\text{$I_{T-1}(v)$ is not good}] \\
    & \geq 1 - \frac{p}{r}.
  \end{align*}
  If the neighborhood is not good, the conditional failure probability of $A$ is at least $r$. As the total failure probability is at most $p$, we get the above. We can lower bound the success probability of black nodes in $A'$ by observing that
  \begin{equation} \label{eq:success-simulation}
      \Pr[\text{$A'$ succeeds at $v$}] \geq \Pr[\text{$I_{T-1}(v)$ is good}](1 - \Pr[\text{$A'$ fails at $v$} \mid \text{$I_{T-1}(v)$ is good}]).
  \end{equation}

  Assume that the $(T-1)$-neighborhood of a black node $v$ is good. Let $u_1, \dots, u_{\delta}$ be the neighbors of $v$. Assume that $A'$ fails at node $v$, i.e., assume that there exist $\sigma_1 \in S_v(u_1), \sigma_2 \in S_v(u_2), \dots, \sigma_{\delta} \in S_v(u_{\delta})$ such that the multiset $\{ \sigma_1, \dots, \sigma_{\delta} \}$ is not contained in the black constraint $B$ of $\Pi$. This implies that there is a significant probability for $A$ to fail in this neighborhood as well. Since the labels $\sigma_1, \sigma_2, \dots, \sigma_\delta$ appear independently and are all frequent, the probability having $A(u_i,v) = \sigma_i$ for all $i$ given $I_{T-1}(v)$ is at least $f^\delta$. We have that
  \[
    \Pr[\text{$I_{T-1}(v)$ is good}] \cdot \Pr[\text{$A'$ fails at $v$} \mid \text{$I_{T-1}(v)$ is good}] \cdot f^{\delta} \leq \Pr[\text{$A$ fails at $v$}].
  \]
  Since the local failure probability of $A$ is at most $p$, we have that
  \begin{align*}
    \Pr[\text{$A'$ fails at $v$} \mid \text{$I_{T-1}(v)$ is good}] &\leq \frac{\Pr[\text{$A$ fails at $v$}]}{\Pr[\text{$I_{T-1}(v)$ is good}] \cdot f^\delta} \\
    &\leq \frac{p}{(1-p/r) \cdot f^\delta}.
  \end{align*}

  By \eqref{eq:success-simulation} we get that black nodes succeed in $A'$ with probability at least
  \[
    \Bigl(1 - \frac{p}{r}\Bigr)\Bigl(1 - \frac{p}{(1-p/r) \cdot f^\delta}\Bigr) = 1 - \frac{p}{r} - \frac{p}{f^\delta}.
  \]
 By choosing $r = f^\delta$, we obtain that
  \[
    \Pr[\text{$A'$ succeeds at a black node}] \geq 1 - \frac{2p}{f^\delta}.
  \]

  Next we consider the white nodes in the simulation algorithm $A'$. According to the definition of $\Pi'$, each white node $u$ succeeds if and only if the sets algorithm $A'$ outputs on $u$'s incident edges allow for selecting one label (from $\Sigma$) from each set such that the resulting selected multiset is a configuration in $W$. This is guaranteed to happen when the true outputs $A(u,v)$ of $u$ are included in the corresponding sets $A'(v,u)$, that is, they are all frequent.

  We begin by upper bounding the probability that a true output is not frequent. Consider a white node $u$, and its black neighbor $v$. From the definition of a frequent label it follows that, for (any) fixed $I_{T-1}(v)$ and any label $\sigma \in \Sigma$, the conditional probability that $\sigma = A(u,v)$ but $\sigma$ is not frequent is at most $f$. We obtain that
  \[
    \Pr[\text{$A(u,v) \notin A'(v,u)$}] \leq |\Sigma|f
  \]
  by a union bound. A white node can fail in $A'$ if $A$ would have failed in the neighborhood or if at least one edge does not have the true output as a frequent output. Taking another union bound over all the edges, we have that a white node succeeds with probability at least $1 - p - d |\Sigma| f$.

  It remains to set $f$ to minimize the local failure probability of the simulation. First, note that we must have $f \leq 1 / |\Sigma|$ to guarantee that there is a frequent output. We choose $f$ to upper bound the maximum of $p + d |\Sigma| f$ and $2p / f^\delta$. By choosing for example $f = Cp^{1/(\delta+1)}$ for $C = (2 / (d |\Sigma|))^{1/(\delta+1)}$, we get for black nodes that
  \[
    \Pr[\text{$A'$ fails at a black node}] \leq \frac{2p}{f^{\delta}} = 2p C^{-\delta} p^{-\delta / (\delta+1)} = 2C^{-\delta}p^{1/(\delta+1)},
  \]
  and for white nodes that
  \[
    \Pr[\text{$A'$ fails at a white node}] \leq d|\Sigma|f + p = d|\Sigma|Cp^{1/(\delta+1)} + p.
  \]

  By the choice of $C$ we have that both terms are bounded by
  \[
    2^{1/(\delta+1)}(d|\Sigma|)^{\delta / (\delta + 1)} p^{1/(\delta+1)} + p \leq Kp^{1/(\delta+1)}
  \]
  for some constant $K = K(\delta,d,|\Sigma|)$.
\end{proof}

\subsection{\fdso{} problem} \label{ssec:problem definition}

Our proof employs an intermediate (non-binary) LCL problem, which we refer to as \emph{forbidden degree or sinkless orientation} (FDSO). As we show in Lemma~\ref{lem:fdso-fd-so}, this problem is at least as easy as both the forbidden degree problem and bipartite sinkless orientation, i.e., as easy as the two problems that play important roles in Lemma~\ref{lem:classificationLB}. It is defined as follows in our white-black formalism.

\begin{definition} \label{def-p0}
Forbidden degree $s$ or sinkless orientation (\sfdso{}), for $0 < s < d$ is a bipartite labeling problem $(\Sigma,d,\delta,W,B)$ over four labels: $\Sigma = \{X,H,T,A\}$. The white constraint $W$ is given by the configurations
\begin{align*}
    &AX^{d-1}\\
    &H^{s+1}X^{d-s-1}\\
    &T^{d-s+1}X^{s-1},
\end{align*}
and the black constraint $B$ is given by the configurations
\begin{align*}
    &X[AHTX]^{\delta - 1}\\
    &HT[AHTX]^{\delta - 2}.
\end{align*}
\end{definition}

\begin{lemma} \label{lem:fdso-fd-so}
	\sfdso{} is at least as easy as the bipartite sinkless orientation problem (for the same $d, \delta$), and at least as easy as the forbidden degree problem (for the same $d, \delta, s$).
\end{lemma}

\begin{proof}
Fix $d, \delta \geq 2$ and $0 < s < d$ arbitrarily. We first show that \sfdso{} is at least as easy as bipartite sinkless orientation, i.e., as the problem with $W = \p\ppp\n$ and $B = \n\p\ppp$. Given a solution to bipartite sinkless orientation, a white node labels one 0-edge with $A$, and the remaining edges with $X$. This is clearly a valid solution for \sfdso{}: each white node is labeled $AX^{d-1}$ and each black node has a 1-edge that is labeled $X$.

Now, we show that \sfdso{} is at least as easy as the forbidden degree problem (with parameters $d, \delta, s$). Assume we are given a solution to $\Pi=(d,\delta,W',B')$ for $W' = [d] \setminus \{ s \}$, and $B' = \{ 1,2,\dots,\delta-1 \}$. In the solution to $\Pi$, white nodes have either at most $s-1$ or at least $s+1$ incident edges labeled $1$. In the first case, they will select $d-s+1$ edges among the 0-labeled edges and output $T$ on those, and $X$ on the rest. In the second case, they will select $s+1$ edges among the 1-labeled edges and output $H$ on these, and $X$ on the rest. A white node is always happy according to \sfdso{}. Assume that a black node does not have any incident label $X$, and therefore does not satisfy the first set of configurations $X[AHTX]^{\delta-1}$. Then, it must have incident edges labeled 0 and 1 in $\Pi$, which must be labeled $T$ and $H$ in \sfdso{}, respectively, and it follows that it satisfies the second set of configurations $HT[AHTX]^{\delta-2}$.
\end{proof}

\subsection{\fdso{} is a fixed point} \label{ssec:fsdo-lb}

In this section we show that, for any $s$, \sfdso{} is a fixed point. In particular, after two round elimination steps starting from \sfdso{}, the output problem is again \sfdso{}.

\begin{lemma} \label{lem:speedup-part-1}
Let $\Pi_0 = \sfdso{}$ for some $0 < s < d$. Then $\Pi_1 = (2^\Sigma,d,\delta,W_1,B_1)$, the black output problem of $\Pi_0$, is the following problem. Given the renaming of labels as
\begin{align*}
\bA &=\set{X}\\
\bH &=\set{H,X}\\
\bT &=\set{T,X}\\
\bX &=\set{A,H,T,X},
\end{align*}
the white constraint $W_1$ is given by the configurations
\begin{align*}
  &\bX [\bA \bH \bT \bX]^{d-1}\\
  &\bH^{s+1}[\bA \bH \bT \bX]^{d-s-1}\\
  &\bT^{d-s+1}[\bA \bH \bT \bX]^{s-1},
\end{align*}
and the black constraint $B_1$ by the configurations
\begin{align*}
  &\bA \bX^{\delta-1}\\
  &\bH \bT \bX^{\delta-2}.
\end{align*}
\end{lemma}

\begin{proof}
  According to the definition of a black output problem, we have to show that
\begin{enumerate}[label=(\arabic*)]
  \item the specified black constraint $B_1$ contains exactly those maximal configurations $\bY_1 \dots \bY_{\delta}$ for which for every choice $Y_1 \in \bY_1, \dots, Y_\delta \in \bY_{\delta}$, we have that $Y_1 \dots Y_\delta$ is a black configuration for $\Pi_0$, and
  \item the specified white constraint $W_1$ contains exactly those configurations $\bZ_1 \dots \bZ_{\delta}$  for which there exists a choice $Z_1 \in \bZ_1, \dots, Z_\delta \in \bZ_{\delta}$ such that $Z_1 \dots Z_\delta$ is a white configuration for $\Pi_0$.
\end{enumerate}

  We start by proving (1). Recall that the black constraint of $\Pi_0$ is given by $X[AHTX]^{\delta - 1}$ and $HT[AHTX]^{\delta - 2}$. Consider the first configuration $\bA \bX^{\delta-1}$ in the black constraint $B_1$ of $\Pi_1$. Clearly, choosing one original label (i.e., a label from $\Sigma$) from $\bA$, and one from each of the $\delta-1$ copies of $\bX$ will result in a configuration that is contained in $X[AHTX]^{\delta - 1}$ since the only choice to pick from $\bA$ is $X$. Similarly, for the second configuration $\bH \bT \bX^{\delta-2}$ in $B_1$, any choice of labels (from $\Sigma$) from $\bH$, $\bT$, and $\delta-2$ copies of $\bX$ yields a configuration contained in $X[AHTX]^{\delta - 1}$ (if at least one $X$ is chosen), or a configuration contained in $HT[AHTX]^{\delta - 2}$ (if no $X$ is chosen, which implies that the two labels chosen from $\bH$ and $\bT$ are $H$ and $T$).

  What is left to be shown is that $\bA \bX^{\delta-1}$ and $\bH \bT \bX^{\delta-2}$ are the only maximal configurations. (It is straightforward to check that these configurations are indeed maximal by checking that any new element added to any of the sets in one of these configurations will break the property given in the definition of the black constraint of the black output problem.) Now assume that there is another maximal configuration $\bY_1 \dots \bY_\delta$. Then, due to its maximality, each of the $\bY_i$ has to contain at least one label $Y_i$ such that $Y_i \neq X$ (as otherwise $\bY_1 \dots \bY_\delta$ would be ``dominated'' by $\bA \bX^{\delta-1}$), and, using this fact, we can also infer that there is a $\bZ \in \{ \bH, \bT \}$ such that, for all $i$, we have $\bY_i \setminus \bZ \neq \emptyset$ (as otherwise $\bY_1 \dots \bY_\delta$ would be dominated by $\bH \bT \bX^{\delta-2}$). From this, it follows that we can pick $Y_1 \in \bY_1, \dots, Y_\delta \in \bY_\delta$ such that $X \neq Y_i \neq H$ for all $i$ (if $\bZ = \bH$), or $X \neq Y_i \neq T$ for all $i$ (if $\bZ = \bT$). But this implies that $Y_1 \dots Y_\delta$ is contained neither in $X[AHTX]^{\delta - 1}$, nor in $HT[AHTX]^{\delta - 2}$, yielding a contradiction. Hence, $\bA \bX^{\delta-1}$ and $\bH \bT \bX^{\delta-2}$ are the only maximal configurations. As $\bA$, $\bH$, $\bT$, and $\bX$ are the only labels that occur in the black configurations of $\Pi_1$, we can restrict attention to these labels also when we discuss the white constraints next.

  To show (2), consider the white configurations in Definition~\ref{def-p0}.
  In the following, we list all white configurations of the black output problem of $\Pi_0$, by simply going through the three white configurations of $\Pi_0$ one by one, each time listing all ``superconfigurations'' consisting of $\delta$ labels from $\{ A, H, T, X \}$ (i.e., those configurations from which the considered white configuration of $\Pi_0$ can be picked).
  For the first configuration of $\Pi_0$, this means anything containing an $A$ and anything containing $X$ (i.e.\ any label) in the remaining positions. We obtain the set of configurations $\{ A, H, T, X \} [\{ X \} \{H, X \} \{T, X \} \{ A, H, T, X \}]^{d-1}$ which is $\bX [\bA \bH \bT \bX]^{d-1}$ after renaming. For the second configuration, we require $s+1$ new labels containing $H$ and the remaining labels containing $X$. We get the set of configurations
  \[
  [\{H,X\} \{A,H,T,X\}]^{s+1} [\{ X \} \{H, X \} \{T, X \} \{ A, H, T, X \}]^{d-s-1}
  \]
  which is $[\bH, \bX]^{s+1} [\bA \bH \bT \bX]^{d-s-1}$ after renaming. We note that any contained configuration containing the label $\bX$ is already contained in the first obtained set of configurations, and therefore we can simplify the second set to the claimed $\bH^{s+1} [\bA \bH \bT \bX ]$. The last white configuration of $\Pi_0$ is similar to the second. We require $d-s+1$ new labels containing the original label $T$ and the remaining to contain the label $X$. We get the set of configurations
  \[
  [\{T,X\} \{A,H,T,X\}]^{d-s+1} [\{ X \} \{H, X \} \{T, X \} \{ A, H, T, X \}]^{s-1}
  \]
  which is $[\bT, \bX]^{d-s+1} [\bA \bH \bT \bX]^{s-1}$ after renaming. Again, any configuration containing the label $\bX$ is already contained in the first obtained set of configurations, so we can simplify the third set to $\bT^{d-s+1} [\bA \bH \bT \bX]^{s-1}$.
As we can see, the obtained white configurations are exactly those listed in the lemma, which completes the proof.
\end{proof}

We proceed to prove that the white output problem of $\Pi_1$ is $\sfdso{}$.

\begin{lemma} \label{lem:speedup-part-2}
  The white output problem $\Pi_2$ of $\Pi_1$, from Lemma~\ref{lem:speedup-part-1}, is \sfdso{}. In particular, $\Pi_2$ is the following problem. Given a renaming of the labels as
  \begin{align*}
    A &= \{ \bX \} \\
    H &= \{ \bH, \bX \} \\
    T &= \{ \bT, \bX \} \\
    X &= \{ \bA, \bH, \bT, \bX \},
  \end{align*}
  the white constraint $W_2$ of $\Pi_2$ is given by the configurations
  \begin{align*}
    &AX^{d-1}\\
    &H^{s+1}X^{d-s-1}\\
    &T^{d-s+1}X^{s-1},
  \end{align*}
  and the black constraint $B_2$ of $\Pi_2$ is given by the configurations
  \begin{align*}
    &X[AHTX]^{\delta - 1}\\
    &HT[AHTX]^{\delta - 2}.
  \end{align*}
\end{lemma}

\begin{proof}
  The proof proceeds similarly to the proof of Lemma~\ref{lem:speedup-part-1}.
  According to the definition of a white output problem, we have to show that
\begin{enumerate}[label=(\arabic*)]
  \item the specified white constraint $W_2$ contains exactly those maximal configurations $Y_1 \dots Y_d$ for which for every choice $\bY_1 \in Y_1, \dots, \bY_d \in Y_d$, we have that $\bY_1 \dots \bY_d$ is a black configuration for $\Pi_1$, and
  \item the specified black constraint $B_2$ contains exactly those configurations $Z_1 \dots Z_d$  for which there exists a choice $\bZ_1 \in Z_1, \dots, \bZ_d \in Z_d$ such that $\bZ_1 \dots \bZ_d$ is a white configuration for $\Pi_1$.
\end{enumerate}

  We start by proving (1). Recall that the white constraint of $\Pi_1$ is given by $\bX [\bA \bH \bT \bX]^{d-1}$, $\bH^{s+1}[\bA \bH \bT \bX]^{d-s-1}$, and $\bT^{d-s+1}[\bA \bH \bT \bX]^{s-1}$. Consider the first configuration $AX^{d-1}$ in the white constraint $W_2$ of $\Pi_2$. Clearly, choosing one label from $A$, and one from each of the $d-1$ copies of $X$ will result in a configuration that is contained in $\bX [\bA \bH \bT \bX]^{d-1}$ since the only choice to pick from $A$ is $\bX$. Similarly, for the second configuration $H^{s+1}X^{d-s-1}$ in $W_2$, any choice of $s+1$ labels from $H$ and $d-s-1$ labels from $X$ yields a configuration contained in $\bX [\bA \bH \bT \bX]^{d-1}$ (if at least one $\bX$ is chosen), or a configuration contained in $\bH^{s+1}[\bA \bH \bT \bX]^{d-s-1}$ (if no $\bX$ is chosen, which implies that all $s+1$ labels chosen from $H$ are $\bH$). Finally, for the third configuration $T^{d-s+1}X^{s-1}$ in $W_2$, any choice of $d-s+1$ labels from $T$ and $s-1$ labels from $X$ yields a configuration contained in $\bX [\bA \bH \bT \bX]^{d-1}$ (if at least one $\bX$ is chosen), or a configuration contained in $\bT^{d-s+1}[\bA \bH \bT \bX]^{s-1}$ (if no $\bX$ is chosen, which implies that all $d-s+1$ labels chosen from $T$ are $\bT$).

  What is left to be shown is that $AX^{d-1}$, $H^{s+1}X^{d-s-1}$, and $T^{d-s+1}X^{s-1}$ are the only maximal configurations. (Again, it is straightforward to check that these configurations are indeed maximal.) Now assume that there is another maximal configuration $Y_1 \dots Y_d$. Then, due to its maximality, each of the $Y_i$ has to contain at least one label $\bY_i$ such that $\bY_i \neq \bX$ (as otherwise $Y_1 \dots Y_d$ would be ``dominated'' by $AX^{d-1}$), at least $d-s$ many of the $Y_i$ have to contain a label that is not contained in $H$ (as otherwise $Y_1 \dots Y_d$ would be dominated by $H^{s+1}X^{d-s-1}$), and at least $s$ many of the $Y_i$ have to contain a label that is not contained in $T$ (as otherwise $Y_1 \dots Y_d$ would be dominated by $T^{d-s+1}X^{s-1}$). Now we claim that we can pick $\bY_1 \in Y_1, \dots, \bY_d \in Y_d$ such that $\bY_i \neq \bX$ for all $i$, $\bY_i \neq \bH$ for $d-s$ many $i$, and $\bY_i \neq \bT$ for $s$ many $i$ as follows.

  Let $J_H \in [d]$ be a set of indexes such that $Y_i \setminus H \neq \emptyset$ for all $i \in J_H$ and $|J_H| = d - s$, and let $J_T \in [d]$ be a set of indexes such that $Y_i \setminus T \neq \emptyset$ for all $i \in J_T$ and $|J_T| = s$. These index sets exists due to our observations above. Now for all $i \in J_H \setminus J_T$, pick a $\bY_i$ satisfying $\bX \neq \bY_i \neq \bH$; for all $i \in J_T \setminus J_H$, pick a $\bY_i$ satisfying $\bX \neq \bY_i \neq \bT$; and for all $i \notin J_H \cup J_T$, pick a $\bY_i$ satisfying $\bX \neq \bY_i$. Since $H \cap T = \{ \bX \}$, we can see that for each $i \notin J_H \cup J_T$, we have that $\bY_i$ is not contained in $H$ or not contained in $T$. Furthermore, at least $(d-s)-|J_H \cap J_T|$ of the $d-|J_H \cap J_T|$ many $\bY_i$ we fixed so far are not contained in $H$, and at least $s-|J_H \cap J_T|$ are not contained in $T$. Since $d=(d-s) + s$, it follows that we can fix the remaining $|J_H \cap J_T|$ choices $\bY_i$ (where $i \in J_H \cap J_T$) so that at least $(d-s)$ of all $\bY_i$ are not contained in $H$, and at least $s$ of all $\bY_i$ are not contained in $T$. This concludes the proof of the claim.

  The choice of our $\bY_i$ immediately implies that $\bY_1 \dots \bY_d$ is not a configuration contained in the white constraint $W_1$ of $\Pi_1$, yielding a contradiction to our assumption. Hence, $AX^{d-1}$, $H^{s+1}X^{d-s-1}$, and $T^{d-s+1}X^{s-1}$ are the only maximal configurations.

  To show (2), consider the black configurations in $\Pi_1$. In the following, we list all black configurations of the white output problem of $\Pi_1$, by going through the two black configurations of $\Pi_1$ one by one, analogously to the approach in the proof of Lemma~\ref{lem:speedup-part-1}.

  The first black configuration requires one element containing the label $\bA$ and $\delta-1$ elements containing the label $\bX$. This gives the set of configurations $X [A H T X]^{\delta-1}$ after renaming. The second black configuration requires an element containing $\bH$, an element containing $\bT$, and the remaining elements to contain $\bX$. This gives the set of configurations $[HX][TX][AHTX]^{\delta-2}$ after renaming. We observe that anything containing the element $X$ is already covered by the first set of configurations, so we can simplify the second set of configurations to $HT[AHTX]^{\delta-2}$. As we can see, the obtained black configurations are exactly those listed in the lemma, which completes the proof.
\end{proof}

\subsection{\fdso{} is nontrivial} \label{sssec:fdso-nontrivial}

It remains to show that the \fdso{} problem cannot be solved in zero rounds.

\begin{lemma} \label{lem:fdso-zero-rounds}
  Any 0-round randomized algorithm for \sfdso{} has local failure probability $p \geq C$ for some constant $C$.
\end{lemma}

\begin{proof}
  An \sfdso{} algorithm fails at a black node if it is labeled with a configuration from $[AH]^\delta$ or $[AT]^\delta$. This is because the first configuration requires an $X$ and the second configuration requires both an $H$ and a $T$.

	A white 0-round algorithm assigns probabilities $p_1$, $p_2$, and $p_3$ to outputting the three white configurations $AX^{d-1}$, $H^{s+1}X^{d-s-1}$, and $T^{d-s+1}X^{s-1}$. If $1 - p_1 - p_2 - p_3 = \Omega(1)$, then it outputs an illegal output with a constant probability, and we are done. If not, we have that at least one of $p_i$ is $\Omega(1)$. We only consider the case $p_1 = \Omega(1)$, as the other cases are similar.

  Consider an arbitrary black node $u$. Each incident white node $v$ chooses the configuration $AX^{d-1}$ independently with probability $p_1$. Since the white nodes can see only their random bits, they will assign the label $A$ to a port with probability at least $1/d$. We label each edge $\{v,u\}$ with that port from the side of $v$. We get that the incident edges of $u$ are labeled with the illegal configuration $A^{\delta}$ with probability at least $(p_1 / d)^{\delta}$, which is $\Omega(1)$, as required.
\end{proof}

\subsection{Logarithmic lower bound for \fdso{}}

We are now ready to put the pieces together and prove that any forbidden degree problem and bipartite sinkless orientation require $\Omega(\log \log n)$ randomized time.

\begin{theorem} \label{thm:forbiddenDegree}
  All forbidden degree problems have randomized complexity $\Omega(\log \log n)$ in the LOCAL model. The same holds for bipartite sinkless orientation.
\end{theorem}

\begin{proof}
  By Lemma~\ref{lem:fdso-fd-so}, it suffices to prove the desired lower bound for \sfdso{}. For the sake of contradiction, assume that there exists an algorithm $A$ with running time $2T = o(\log \log n)$ for solving \sfdso{} with probability $1-1/n$. This implies that $A$ solves \sfdso{} with local failure probability at most $1/n$.

  Using Lemmas~\ref{lem:speedup-part-1} and~\ref{lem:speedup-part-2}, we have that after two round elimination steps, starting with \sfdso{}, the obtained problem is again \sfdso{}. Assuming that \sfdso{} can be solved with local failure probability $p$ in time $T'$ and given Lemma~\ref{lem:rand-generalized-simulation}, we get that \sfdso{} can be solved in time $T'-2$ with local failure probability at most
  \[
  p' \leq K_2\Bigl(K_1 p^{1/(d+1)}\Bigr)^{1/(\delta+1)},
  \]
  where $K_1 = K(d,\delta,4)$ and $K_2 = K(\delta, d, 4)$. Let $\Delta = \max \{ d, \delta \}$. We have that
  \[
  p' \leq K_2 K_1^{1/(\delta+1)} \Bigl(p^{1/(\Delta+1)}\Bigr)^{1/(\Delta+1)} \leq K_2 K_1^{1/(\delta+1)}p^{1/(\Delta+1)^2}.
  \]

  Starting with the assumed $2T$-time algorithm and iterating the above step $T$ times we arrive at a 0-round algorithm for \sfdso{} with local failure probability at most some $r$. Since for some large enough constant $C$ we have that $K_2 C^{1/(\delta+1)} \leq C$ and $K_1 C^{1/(d+1)} \leq C$, we get that the multiplicative term in front of $p$ converges to some constant $K$. We get the bound
  \[
    r \leq K \biggl(\frac{1}{n}\biggr)^{1/(\Delta+1)^T}.
  \]
  If $T = o(\log \log n)$, then $r = o(1)$. By Lemma~\ref{lem:fdso-zero-rounds} we know that for any 0-round algorithm $r = \Omega(1)$, a contradiction with the assumption that $T = o(\log \log n)$.
\end{proof}

\begin{corollary} \label{cor:forbiddenDegree}
  All forbidden degree problems require $\Omega(\log n)$ deterministic time in the LOCAL model. The same holds for bipartite sinkless orientation.
\end{corollary}

\begin{proof}
  By Theorem~\ref{thm:forbiddenDegree}, the deterministic complexity of any forbidden degree problem is $\Omega(\log \log n)$. Since there are no LCLs with deterministic complexity between $\omega(\log^* n)$ and $o(\log n)$~\cite{chang16exponential}, the complexity must be $\Omega(\log n)$.
\end{proof}

\subsection{Proof of Theorem~\ref{thm:log-lower}}

Lemma~\ref{lem:classificationLB} shows that any problem that is not of type $\tI$, $\tII$, $\tIII$, $\tIV$, $\tV$ or $\tVI$ can be relaxed to a forbidden degree problem, the bipartite sinkless orientation problem, or a problem equivalent to one of those. Corollary~\ref{cor:forbiddenDegree} shows that all forbidden degree problems and bipartite sinkless orientation require $\Omega(\log n)$ rounds.
Observation~\ref{obs:equiv} ensures that this lower bound also holds for all problems that are equivalent to a forbidden degree problem or bipartite sinkless orientation.
Hence, Theorem~\ref{thm:log-lower} follows.

\section{Randomized complexity}\label{sec:rand}

We have now completed the classification of the deterministic distributed complexity of binary labeling problems in trees; all solvable problems fall in one of the complexity classes of $O(1)$, $\Theta(\log n)$, and $\Theta(n)$. Trivially, $O(1)$ deterministic complexity also implies $O(1)$ randomized complexity, and it is straightforward to verify that all problems discussed in Section~\ref{sec:global} require $\Theta(n)$ rounds not only for deterministic algorithms but also for randomized algorithms. However, the class of $\Theta(\log n)$ is much more interesting, as it contains both problems in which randomness helps and also problems in which randomness does not help. In this section we will explore this in more detail and show how to classify \emph{some} binary labeling problems according to their randomized complexity.

For example, recall the problem of even orientation (Example~\ref{ex:even-orientation}) in $3$-regular graphs, defined by $W=\p\n\p\n$, $B=\n\p\n$. In this section, we prove in Theorem~\ref{thm:propagation} that even orientation has randomized complexity $\Omega(\log n)$. Combining this with the fact that even orientation is $(2,1)$-resilient and thus by Theorem~\ref{theorem21resil} has a deterministic algorithm with complexity $O(\log n)$, we can deduce that even orientation in $3$-regular graphs has randomized complexity $\Theta(\log n)$.

Throughout the section, it is assumed that $d\geq 3$, or $\delta \geq 3$.

\subsection{Propagation lower bound}
In this section, we prove that if $\Pi=(d,\delta,W,B)$ is a binary LCL with deterministic complexity $\Theta (\log n)$ such that both white and black constraints do not have a so called \emph{escape}, then also the randomized complexity of $\Pi$ is $\Theta(\log n)$. We begin with the definition of an escape using the vector notation (cf.\ Section~\ref{ssec:binary-labeling-problems}).
\begin{definition}[Escape]
Given a binary LCL $\Pi=(d,\delta,W,B)$ we say that $\Pi$ has a \emph{white} (resp.\ \emph{black}) \emph{escape} if $W$ (resp.\ $B$) contains two consecutive $\p$s.
\end{definition}

Note that problem of even orientation (Example~\ref{ex:even-orientation}) mentioned in the beginning of the section has neither black nor white escape.

The definition of an escape immediately implies the following observation.
\begin{observation}[Propagation]
\label{lem:propagation}
Let $\Pi$ be a binary LCL with no white and no black escape.
If we fix for a given node $v\in V$ the labeling of all but one port in any way, then there is at most a single way to complete the labeling of the last port to produce a valid output for node $v$.
\end{observation}

\begin{theorem}[Randomized propagation lower bound]\label{thm:propagation}
Let $\Pi=(d,\delta,W,B)$ be a binary LCL such that there is neither a white nor a black escape and assume that $\Pi$ has deterministic complexity $\Theta(\log n)$ on trees. Then $\Pi$ has randomized complexity $\Theta(\log n)$ on trees.

\end{theorem}
\begin{proof}
Let $\Pi$ be a problem satisfying the conditions of the theorem and let $d$ denote the degree of white nodes and let $\delta$ be the degree of black nodes (note that the considered trees are regular except for the degrees at the leaves). We emphasize that the degrees $d$ and $\delta$ are fixed and constant throughout this proof. As there is a deterministic $O(\log n)$-round algorithm for $\Pi$ there is also a randomized algorithm that solves $\Pi$ in $O(\log n)$ rounds.

Next, we show that there is no randomized algorithm for $\Pi$ with complexity $o(\log n)$. Assume towards a contradiction that there is an algorithm $\mathcal{A}$ solving $\Pi$ in $o(\log n)$ rounds and fix a constant $0<c<(100 \log \max\{d,\delta\})^{-1}$. By the assumption there exists some $n_0$ such that $\mathcal{A}$ solves $\Pi$ on any tree with $n_0$ nodes (that satisfies the degree restrictions of $\Pi$) in less than $k(n_0)=c\cdot \log n_0$ rounds. Fix any $n$ larger than $n_0$ for which a tree meeting the degree requirements of $\Pi$ exists. We show that there is a tree with $n$ nodes and a port assignment on which algorithm $\mathcal{A}$ has to fail.
To this end let $T=(V,E)$ be  a tree (satisfying the degree restrictions of $\Pi$) with $n$ nodes such that it has an internal white node $u$ that is more than $10 k $ hops away from any leaf of $T$. Due to the choice of $c$ such a tree exists. However, constants in this proof are not optimized to ease presentation.
Denote by $v_i,1\leq i\leq d$ the neighbors of $u$, by $e_i=\{u,v_i\}$ the incident edges of $u$ and by $T_i=(V_i,E_i)$ the tree of depth $2k-1$ rooted at $v_i$ such that $u\not\in V_i$.
For each tree $T_i$, denote by $S_i\subseteq V_i$ the subset of nodes that are at distance $2k$ from $u$, that is, $S_i$ is the set of leaves of $T_i$. Let $F_i$ be the set of edges of $T_i$ that are incident to a node in $S_i$. Equip the tree with any port labeling such that \emph{corresponding nodes} in $S_i$ and $S_j$ for $1\leq i,j\leq d$ have the exact same $k$-hop view.
Since both white and black constraints have no escape, we have full propagation in the following sense:

\medskip
\noindent\textbf{Propagation Claim:} For any $1\leq i\leq d$ and under the assumption that $\mathcal{A}$ is a correct algorithm for $\Pi$ the output of $\mathcal{A}$ on the edges of $F_i$ uniquely determines the output labels of all other edges of $T_i$ and also the output label of the edge $e_i$.

\medskip
\noindent Let us prove the propagation claim. If we fix for a given node $v$ the labeling of all but one ports in any way, there is at most a single way to complete the labeling of the last port to produce a valid output for $v$ (cf.\ Observation \ref{lem:propagation}). The claim follows by an induction on the levels of the tree $T_i$ starting with the `leaf edges' incident to $S_i$. Hence the claim follows.

As $\mathcal{A}$ runs in less than $k$ rounds for all $1\leq i\leq d$ and the nodes in $S_i$ have distance $2k$ from $u$, the nodes in $S_i$ decide on their output independently from the information in the graph $T_j$ for all $j\neq i$ and also independent from $e_j$. Even stronger this independence is not just pairwise but the output on the sets of edges $F_1,\ldots,F_d$ are fully independent. Furthermore, as the graph and the port numbering around all these sets is identical the outputs produced by the randomized algorithm $\mathcal{A}$ on $F_1,\ldots,F_d$ follow the same probability distribution. Using the propagation claim we obtain that also all edges $e_1,\ldots,e_d$ follow the same output distribution and the output on these edges is fully independent.
Thus all edges $e_1,\ldots,e_d$ are in one of three following (non exclusive) cases: For all $1\leq i\leq d$
\begin{itemize}
\item \textbf{Case 1:} the output of $S_i$ forces $e_i$ to be labeled $0$ with probability at least $1/4$ and it forces $e_i$ to be labeled with $1$ with probability at least $1/4$, or
    \item \textbf{Case 2:} with probability at least $1/4$, the output of the nodes in $S_i$ forces $e_i$ to be labeled with $1$, or
    \item \textbf{Case 3:} with probability at least $1/4$, the output of the nodes in $S_i$ forces $e_i$ to be labeled with $0$.
\end{itemize}

We deduce a contradiction for each case using that all edges are in the same case and the problem $\Pi$ is non trivial as it has deterministic complexity $\Theta(\log n)$ by assumption.

\medskip

\textbf{Case 1:}
As $\Pi$ is not a trivial problem there exists some $m\leq d$ such that $\Pi$ does not allow for exactly $m\leq d$ ports to be labeled $1$ incident around white nodes. However, as each edge $e_1,\ldots,e_d$ is independently labeled with $1$ or $0$ each  with probability at least $1/4$ the probability that the first $m$ edges $e_1,\ldots,e_m$ incident to $u$ are labeled with $1$ and the remaining incident edges $e_{m+1},\ldots e_d$ are labeled with $0$ is constant, i.e., it is at least $(1/4)^m\cdot (1/4)^{d-m}=(1/4)^d=\Omega(1)$. As this labeling is not allowed for $u$ the algorithm fails at least with constant probability at $u$, a contradiction.\footnote{The actual probability of failure at a white node in this case is much larger but the presented probability is sufficient for a contradiction.}

\medskip

\textbf{Cases 2 and 3:} Both cases are symmetric with interchanging the role of the labels.
Assume we are in case 2, and first assume $\Pi$ does not allow white nodes to label all incident edges with $1$.
Then with constant probability, i.e., with probability at least $(1/4)^{d}$, algorithm $\mathcal{A}$ labels all edges incident to $u$ with label $1$, i.e., $\mathcal{A}$ fails to produce a valid output for $u$ with constant probability, a contradiction.

If we are in case 2 and $\Pi$ does allow white nodes to label all ports with $1$ we construct a new tree and a port numbering on which algorithm $\mathcal{A}$ fails for a black node with constant probability:
Consider a black node $t$ and denote by $u_1,\dotsc,u_{\delta}$ its white neighbors, now similarly fix the trees and the port numbers rooted at the children of  $u_i,1\leq i\leq \delta$ to be such that the edges incident to $u_i,1\leq i\leq \delta$ are in case 2, i.e.\ labeled $1$ with probability at least $1/4$, except for the edge connected to $t$.\footnote{The constant $c$ is chosen small enough such that we do not run out of nodes in this construction if we cap trees at a large enough radius that does not influence the output at the considered vertices.}
Now since we have full propagation and no escape, and we know that white nodes are allowed to label their ports with all 1, we deduce that also the edges $\{t,u_i\},1\leq i\leq \delta$ must be labeled with 1 with probability at least $1/4$, since by the no escape property, an output of a single port labeled 0 is not allowed. Thus we have that with probability at least $(1/4)^{\delta}=\Omega(1)$, $t$ labels all its edges with $1$. As this cannot be a feasible output, otherwise the problem would be trivial as then white and black nodes would accept all edges to be labeled with $1$, the algorithm fails with constant probability, a contradiction.

\medskip

In either case, we reach a contradiction, thus there is no algorithm that solves $\Pi$ w.h.p.\ in less than $k$ rounds which proves the claim.
\end{proof}

\subsection{Lower bounds using independent sets}
In this section, we prove the following theorem.
\begin{theorem}\label{lowerboundusingindependentset}
Consider the following binary LCL problem $\Pi=(d,2,W,B)$, $W=\p\p\n^{d-2}\p$, $B=\n\p\n$. Then for all $d\geq 20$, this problem has randomized complexity $\Omega (\log n)$.
\end{theorem}

To prove this theorem, we employ the following lemma, whose proof can be found in \cite[Lemma~21]{Alon2010}.
\begin{lemma}\label{highgirthindependencelemma}
  Let $d\geq 3,\epsilon >0$, then there exists a $d$-regular graph $G$ on $n$ vertices such that $G$ has girth $\Omega(\log n)$, and furthermore $\alpha (G)\leq\left( \frac{2\ln d}{d}+\epsilon \right)n$. Here, $\alpha (G)$ denotes the size of the largest independent set in $G$.
\end{lemma}

Now we are ready to prove Theorem~\ref{lowerboundusingindependentset}.
\begin{proof}
   Given some graph $G=(V,E)$ in the black/white formalism, denote by $V_W,V_B$ the sets of black and white nodes respectively. Note that the constraints of the black nodes allow us to view the problem as an orientation problem on white nodes. Let $c:E\to \set{0,1}$ be some labeling to the edges that induces a valid solution to $\Pi$. Define
   \[E'=\bigl\{\{v_1,v_2\}\bigm|v_1,v_2\in V_W,\,\exists b\in V_B,\,\{v_1,b\} \in E,\,\{b,v_2\}\in E\bigr\}.\]
   We can abuse notation and view $c$ as function from  $E'\to \set{{\leftarrow},{\rightarrow}}$. Basically, as an orientation function. This is since given some black node $b\in B$ and his white neighbors $v,u$, assume w.l.o.g.\ that $c(u,b)=1,c(v,b)=0$, then we define that $c(u,v)={\rightarrow}$, and vice versa. Here ${\rightarrow}$ means that $(u,v)$ is oriented towards $v$, ${\leftarrow}$ is defined similarly. Note that each white node still has total degree $d$, thus we are now working with a $d$-regular graph. From here on in, we treat $E'$ as our set of edges, and $W$ as our set of nodes (we refer to them using $E,V$, respectively). Now, $\outdegree(v)$ of a node $v\in V_W$ is the number of edges incident to it oriented away from $v$. Similarly, we define $\indegree(v)$ of a node $v\in V_W$. Denote by $x_0,x_1,x_d$ the number of nodes with outdegree $0,1,d$ respectively.

   Now, since each edge in the graph is oriented, we have that \[\sum_{v\in V}\indegree(v)=\sum_{v\in V}\outdegree(v).\] Or, in other words, since the total degree of each node is $d$:
   \[
   dx_0+(d-1)x_1=x_1+dx_d.
   \]
   Thus $x_d=x_0+\frac{d-2}{d}x_1$. We also know that $x_0+x_1+x_d=n$. From these two equations we continue with two possible cases:
   \begin{itemize}[noitemsep]
       \item $x_d\geq \frac{n}{3}$.
       \item $x_1\geq \frac{n}{3}$, in which case we have that $x_d\geq \frac{d-2}{3d}n$.
   \end{itemize}

  Now since $d\geq 11$, we can simply deduce that $x_d\geq\frac{d-2}{3d}n $ in all cases.

  Denote by $V_d$ all the nodes with outdegree $d$, note that $V_d$ is an independent set and that $|V_d|=x_d$. Thus we can conclude that any $d$-regular graph in which there is a valid solution to $\Pi$, there is an independent set of size at least $\frac{d-2}{3d}n$.

  Now assume towards a contradiction that $\Pi$ is solvable on $d$-regular trees in $o(\log n)$ rounds. In particular, following from an indistinguishability argument, which holds also for randomized algorithms, the same algorithm should solve $\Pi$ on $d$-regular graphs with girth $\Omega(\log n)$. However, according to Lemma \ref{highgirthindependencelemma}, there is a $d$-regular graph with girth $\Omega(\log n)$ and $\alpha (G)\leq\left( \frac{2\ln d}{d}+\epsilon \right)n$, for any $\epsilon >0$, denote such a graph by $G_0$. Furthermore, each valid solution to $\Pi$ in a graph $G$ induces an independent set of size at least $\frac{d-2}{3d}n$ in $G$, notice that for all $d\geq 20$, it holds that $\frac{d-2}{3d}n>\left( \frac{2\ln d}{d}+\epsilon \right)n$. According to our assumptions $\Pi$ will be solved correctly on $G_0$, since it has $\Omega(\log n)$ girth. This is a contradiction to the fact that any valid solution to $\Pi$ induces an independent set of size at least $\frac{d-2}{3d}n$. Thus, our assumption was wrong and no such algorithm exists. In particular, $\Pi$ has randomized complexity $\Omega(\log n)$ as required.
\end{proof}

\subsection{Lower bounds using cuts}
In this section, our goal is to prove the following theorem,
\begin{theorem}\label{lowerboundusingnolargecuts}
Given some integers $r,d,r<\frac{d}{2}$, consider the following binary LCL $\Pi=(d,2,W,B)$, $W=\p^{r+1}\n^{d-2r-1}\p^{r+1}$, $B=\n\p\n$. Then $\Pi$ has randomized complexity $\Omega(\log n)$ whenever $r< \frac{d}{4}-\frac{\sqrt{d-1}}{2}$.
\end{theorem}
To prove the theorem, we will employ the following lemma.
\begin{lemma}\label{Highgirthnolargecut}
  Given an integer $d\geq 3$ and any $\epsilon >0$, there exists a $d$-regular graph $G$ on $n$ vertices of girth $\Omega(\log n)$, and for all $S\subseteq V$, it holds that $\frac{E(S)}{|E|}\leq \frac{1}{2}+\frac{\sqrt{d-1}+\epsilon}{d}$. Here, $E(S)$ is the number of edges in $G$ crossing the cut $S$.
\end{lemma}

The lemma follows from previous work on the second eigenvalue of random $d$-regular graphs.

\begin{proof}[Proof sketch.]
  Friedman~\cite{Friedman2003,Friedman2008} showed that for the smallest eigenvalue $\lambda_n$ of the adjacency matrix of a random $d$-regular graph (in the standard configuration model, for every positive integer $d$) we have that $|\lambda_n| \leq 2\sqrt{d-1}+\varepsilon$, for any $\varepsilon > 0$, with probability $1-O(n^{-c(d)})$ for some $c(d) = \Omega(\sqrt{d})$.

  It is known that the size of maximum cut of a $d$-regular graph is bounded by $(1/2 + |\lambda_n|/(2d))|E|$ (see for example Trevisan~\cite{Trevisan12maxcut}). This implies that most random $d$-regular graphs do not have large cuts. It remains to construct one with large girth.

  We use the standard trick of cutting cycles~(see e.g.\ \cite[Section 2]{Alon2010}). It is known that with probability at least $1/2$ there are at most $2\sqrt{n}$ cycles of length at most $0.5 \log n / \log d$ in a random $d$-regular graph. Therefore there exists such a graph with no large cuts. It is possible to cut each cycle by removing and adding exactly two edges without introducing any new short cycles. In the worst case all of the new edges are cut by some maximum cut of the original graph. We have that the size of the maximum cut in the new graph is bounded by $(1/2 + (2\sqrt{d-1}+\varepsilon)/(2d))dn/2 + 2\sqrt{n}$, which is at most $(1/2 + (\sqrt{d-1}+2\varepsilon)/d$ for large enough $n$.
\end{proof}

Now, we are ready to prove Theorem \ref{lowerboundusingnolargecuts}.
\begin{proof}
   As in the proof of Theorem \ref{lowerboundusingindependentset}, given some graph $G=(V,E)$ in the black/white formalism, denote by $V_W,V_B$ the sets of black and white nodes respectively. Note that the constraints of the black nodes allow us to view the problem as an orientation problem on white nodes. Let $c:E\to \set{0,1}$ be some labeling to the edges that induces a valid solution to $\Pi$. Define
   \[E'=\bigl\{\{v_1,v_2\}\bigm|v_1,v_2\in V_W,\,\exists b\in V_B,\,\{v_1,b\} \in E,\,\{b,v_2\}\in E\bigr\}.\]
   We can abuse notation and view $c$ as function from  $E'\to \set{{\leftarrow},{\rightarrow}}$. Basically, as an orientation function. This is since given some black node $b\in B$ and his white neighbors $v,u$, assume w.l.o.g.\ that $c(u,b)=1,c(v,b)=0$, then we define that $c(u,v)={\rightarrow}$, and vice versa. Here ${\rightarrow}$ means that $(u,v)$ is oriented towards $v$, ${\leftarrow}$ is defined similarly. Note that each white node still has total degree $d$, thus we are now working with a $d$-regular graph. From here on in, we treat $E'$ as our set of edges, and $W$ as our set of nodes (we refer to them using $E,V$, respectively). Now, $\outdegree(v)$ of a node $v\in V_W$ is the number of edges incident to it oriented away from $v$. We define similarly $\indegree(v)$ of a node $v\in V_W$.

   Now note that $\Pi$ is the problem in which every white node has to have outdegree at most $r$, or indegree at most $r$. Given some solution to $\Pi$ denote by $X_0,X_1$ the sets of nodes of outdegree at most $r$ and indegree at most $r$, respectively. Note that $V\backslash X_0=X_1$, thus these two sets define a cut in the graph. Now we want to show that the number of edges leaving $X_0$ into $X_1$ is bounded from below by $(d-2r)|X_0|$. Assume by towards a contradiction that this is not the case. Denote by $E_\ell$ the set of edges incident to a node in $X_0$ that cross the cut $(X_0,X_1)$. Denote by $E_s$ the set of edges in which both endpoints are in $X_0$. Our assumption is that $|E_\ell|<(d-2r)|X_0|$. Note that we have that $d|X_0|=2|E_s|+|E_\ell|$ simply since $d|X_0|$ is the sum of degrees of nodes in $X_0$. From our assumption we can then deduce that $|E_s|>r|X_0|$. Now, we know that for all $v\in X_0$, $\outdegree(v)\leq r$. Thus $\sum_{v\in X_0} \outdegree(v)\leq r|X_0|$. However, this is a contradiction, since by definition of $E_s$ and the fact that each edge is oriented and the fact that $|E_s|>r|X_0|$ we deduce that $\sum_{v\in X_0} \outdegree(v)> r|X_0|$.

   Thus, we know that $|E(X_0,X_1)|\geq (d-2r)|X_0|$. Similarly, we can prove that $|E(X_0,X_1)|\geq (d-2r)|X_1|$. Assume w.l.o.g.\ that $|X_0|\geq |X_1|$. Thus we get that
   \[
   \frac{|E(X_0,X_1)|}{|E|}\geq \frac{|X_0|(d-2r)}{0.5d(|X_0|+|X_1|)}\geq \frac{|X_0|(d-2r)}{d|X_0|}=1-\frac{2r}{d}.
   \]
 Thus, we proved, that in any $d$-regular graph $G=(V,E)$, a solution to $\Pi$ induces a cut $S\subseteq V$ s.t. $\frac{E(S)}{|E|}\geq 1-\frac{2r}{d}$.

 Now, assume towards a contradiction that there exists an algorithm $A$, that runs in $o(\log n)$ rounds, that solves $\Pi$ w.h.p.\ on any $d$-regular tree $G$. Let $G_0$ be the graph promised from Lemma~\ref{Highgirthnolargecut} for a sufficiently small  $\epsilon>0$.

 Now, since $G_0$ has girth $\Omega(\log n)$, from an indistinguishability argument, that also holds for randomized algorithms, it must hold that $A$ also solves $\Pi$ correctly on $G_0$ w.h.p., since for every node $v\in V$, the view of $v$ when executing $A$ is that of a $d$-regular tree. Thus we can deduce that $G_0$ has a cut $S$ induced by $X_0$ and $X_1$ as explained before, of relative size $1-\frac{2r}{d}$. But, by our assumption, it holds that
 \[
 r< \frac{d}{4}-\frac{\sqrt{d-1}}{2}.
 \]
 Thus, we have that
 \[
 1-\frac{2r}{d}> \frac{1}{2}+\frac{\sqrt{d-1}}{d}.
 \]

 Since we treat $d,r$ as constants, this is a contradiction with Lemma \ref{Highgirthnolargecut} for a sufficiently small $\epsilon>0$.

 Thus, we conclude that no such algorithm $A$ exists. Thus any algorithm that solves $\Pi$ correctly w.h.p.\ on $d$-regular trees must have rounds complexity $\Omega (\log n)$. Thus $\Pi$ has randomized complexity $\Omega(\log n)$ whenever  $ r< \frac{d}{4}-\frac{\sqrt{d-1}}{2}$ as required.
 \end{proof}

\section*{Acknowledgments}

This work was supported in part by the Academy of Finland, Grant 314888 (Juho Hirvonen), and by the European Union's Horizon 2020 Research And  Innovation Programme under grant agreement no.\ 755839 (Yuval Efron, Yannic Maus).

\urlstyle{same}
\bibliographystyle{plainnat}
\bibliography{binary-lcls}
\end{document}